\documentclass[11pt]{article}
\usepackage{amsmath, amssymb, amsthm}
\usepackage{graphicx}
\usepackage{color}
\usepackage{subfigure}
\usepackage{lineno}
\newtheorem{thm}{Theorem}
\usepackage{ulem}
\usepackage{epstopdf}
\usepackage{babel}
\usepackage{fontenc}
\usepackage{bigints}
\usepackage{float}
\date{}
\begin{document}

{\LARGE \bf  
\begin{center}
Spatio-temporal pattern formation under varying functional response parametrizations
\end{center}
}

\vspace*{1cm}

\centerline{\bf Indrajyoti Gaine$^a$, Malay Banerjee$^{a,}$\footnote{Corresponding author: malayb@iitk.ac.in}}

\vspace{0.5cm}

\centerline{ $^a$Department of Mathematics and Statistics, Indian Institute of Technology Kanpur, Kanpur - 208016, India}

\vspace{1cm}

\begin{center}
{\bf Abstract}
\end{center}

Enhancement of the predictive power and robustness of nonlinear population dynamics models allows ecologists to make more reliable forecasts about species' long term survival. However, the limited availability of detailed ecological data, especially for complex ecological interactions creates uncertainty in model predictions, often requiring adjustments to the mathematical formulation of these interactions. Modifying the mathematical representation of components responsible for complex behaviors, such as predation, can further contribute to this uncertainty, a phenomenon known as structural sensitivity. Structural sensitivity has been explored primarily in non-spatial systems governed by ordinary differential equations (ODEs), and in a limited number of simple, spatially extended systems modeled by nonhomogeneous parabolic partial differential equations (PDEs), where self-diffusion alone cannot produce spatial patterns. In this study, we broaden the scope of structural sensitivity analysis to include spatio-temporal ecological systems in which spatial patterns can emerge due to diffusive instability. Through a combination of analytical techniques and supporting numerical simulations, we show that pattern formation can be highly sensitive to how the system and its associated ecological interactions are mathematically parameterized. In fact, some patterns observed in one version of the model may completely disappear in another with a different parameterization, even though the underlying properties remain unchanged.

\vspace{1.0cm}

\noindent
{\bf Keywords:} Turing instability; Pattern formation; functional response; structural sensitivity; Stability.

%\vspace{1cm}

\baselineskip 0.25in

\begin{center}
{\LARGE\bf }
\end{center}

%%%%%%%%%%%%%%%%%%%%%%%%%%%%%%%%%%%
\section{Introduction}{\label{Intro}} 
Mathematical models of interacting populations provide a systematic approach to predicting the long-term responses of natural ecosystems to external influences. The accuracy of the prediction mostly depends on the efficiency of the model under consideration. Recent research trends aim to enhance the predictive power and accuracy of ecological models by incorporating more realistic features of the ecosystem. One significant advancement in ecological modeling is the integration of symmetric diffusion into temporal ecological models, which are formulated under the assumption of uniform species distribution across the entire spatial domain. This enhancement allows for a more accurate representation of spatial heterogeneity of species within their habitats. By incorporating spatial dynamics, these models move closer to accurately capturing key ecological factors, such as the spatial variability of species communities and the role of space in determining species persistence or extinction \cite{ecology}. The primary contribution of contemporary research lies in exploring the complete bifurcation structure of spatio-temporal models. Turing's groundbreaking work demonstrated that the interplay between nonlinear reaction kinetics and diffusion can spontaneously generate spatial patterns such as stripes, spots, or labyrinthine structures, independent of any inhomogeneity in initial or boundary conditions \cite{turing1990chemical}. The mechanism responsible for this spontaneous emergence of spatially heterogeneous stationary patterns through biological or chemical interactions is known as diffusive instability. Segel and Jackson \cite{segel1972dissipative} were the first to apply Turing’s idea to explain the occurrence of diffusive instability in an ecological context. Further research has revealed that the coexistence of oscillatory temporal dynamics and diffusive instability can lead to the emergence of spatio-temporal patterns including stationary and dynamic patterns, traveling waves, oscillatory patches, and chaotic fluctuations. The interaction between diffusive instability and temporal oscillations, as well as its impact on the spatio-temporal dynamics of ecological models, has been investigated numerically in \cite{TH3, TH4, TH6, Hu2015PatternFA} and both theoretically and numerically in \cite{TH1, TH2, SONG2016229}. Beyond ecological models, spatio-temporal patterns also arise in network models \cite{MA2016586, 6678316, ZHANG201879}, chemical systems \cite{doi:10.1137/23M1552668, Jianping, PhysRevE.108.034206}, and various other scientific disciplines.

The inherent complexity of ecosystems often limits the realism and tractability of the spatio-temporal ecological models, described using parabolic partial differential equations with nonlinear reaction kinetics. The formulation of these reaction-diffusion equations frequently depends on some assumptions, which may either oversimplify ecological processes and overlook key behaviors of the ecosystem, or become too complex to understand \cite{Demongeot}. In some cases, simpler models are chosen for mathematical tractability, though they rarely compromise ecological realism. Furthermore, a major challenge in making accurate predictions with these ecological models is the uncertainty about the processes involved. These limitations can be addressed to some extent by using advanced experimental technologies and by collecting more comprehensive field data. However, experimental data are often collected on a small spatial scale within controlled laboratory settings, which can lead to potential inaccuracies in model predictions in the natural environment on larger temporal and spatial scales \cite{Arashkevich,MOROZOV201045}. Another challenge in model formulation lies in determining appropriate functions to describe specific ecological interactions. For most ecological interactions, multiple representative functions are available and their relevance lies in either aligning with the system's known properties and assumptions about underlying processes or fitting empirical data. The choice of the functional form by the modeler is typically guided by the specific features of the interaction between the species under consideration \cite{Martin,Barclay,Baumgaertner,Majmudar}. 

A common intuitive presumption in model prediction was if two functional forms, representing the same mechanism, share similar properties and could be parameterized to closely resemble to each other, the models utilizing these functions should exhibit essentially identical behavior. As a result, researchers often opt for either functional form that conveniently fits the available predation data \cite{Mullin}. However Fussmann and Blasius demonstrated that the degree of resource enrichment required to destabilize the community dynamics in the Rosenzweig–MacArthur model is highly sensitive to the mathematical form of the functional response \cite{Fussmann}, even when they appear visually indistinguishable. This phenomenon is referred to as structural sensitivity. Building on this concept, researchers have studied structural sensitivity for various ecological models, including the Leslie–Gower–May model \cite{Wyse2022StructuralSI}, the Bazykin model \cite{ALDEBERT2016163, Aldebert2019}, Hastings-Powell's model \cite{Gaine}, a predator-prey type model describing chemostat experiment \cite{ALDEBERT20181,FLORA201182}. A common assumption in all these studies is that species concentrations are uniformly distributed within their habitat, limiting the analysis to ODE-based models. The present work seeks to extend this investigation by exploring the impact of structural sensitivity in scenarios where species are not uniformly distributed across space.

Recently, Manna et al. studied the structural sensitivity of a spatially explicit Rosenzweig–MacArthur model \cite{MANNA2024134220}. The classical Rosenzweig–MacArthur model is unable to produce spatially heterogeneous patterns with self-diffusion alone \cite{MB}. To address this limitation, they introduced a cross-diffusion term in the predator equation and further explored the structural sensitivity of spatially heterogeneous patterns by replacing the Holling Type II functional response with Ivlev and trigonometric functional responses—alternative formulations that share common properties such as being zero at zero, monotonically increasing, possessing a finite horizontal asymptote, and being concave downward. Their findings revealed the existence of a spatial ghost attractor in the model with the Holling Type II functional response and the absence of labyrinthine patterns in the case of the model with the trigonometric functional response. Additionally, they observed spatially heterogeneous quasi-periodic patterns in the model with the trigonometric functional response, even though linear stability analysis confirmed that the homogeneous steady state remains stable under spatially heterogeneous perturbations.

In contrast to the Rosenzweig–MacArthur model, the Bazykin model \cite{Bazykin}, which incorporates a density-dependent death rate for the predator population, is capable of producing spatially heterogeneous patterns under self-diffusion. Additionally, the Bazykin model support two coexisting steady states, a feature absent in the Rosenzweig–MacArthur model. The main goal of this work is to study the structural sensitivity of temporal and spatio-temporal dynamics of the Bazykin model. We consider a spatially explicit model with Bazykin type reaction kinetics, and study the structural sensitivity by replacing the Holling Type II functional response with  Ivlev functional response. We derive the analytical results with a general form od the prey-dependent functional response meeting specific properties, such as being zero at zero, monotonic increasing, having a finite horizontal asymptote, and being concave down. A key objective of this study is to establish analytical insights that are independent of the specific parametrization of the functional response. Through numerical investigations, we support the analytical findings and illustrate the intricate dynamics of the predator–prey system. Our results highlight processes that remain consistent across different functional responses, while also revealing notable differences in the global dynamics of the temporal model and the resulting spatio-temporal pattern formation.

The paper is organized as follows: Section \ref{MathModel} introduces the structure of the temporal model and specifies the conditions for the functional response. Section \ref{LocBif} provides an analytical presentation of all possible homogeneous steady states, their stability criteria, and the conditions under which these criteria change based on the general functional form with certain properties. In Section \ref{SpatioModel}, the model is extended spatially, establishing the existence of global solutions and non-stationary coexisting solutions. This section also examines the stability of homogeneous steady states and derives the criteria for Turing instability of the coexisting homogeneous steady state. Section \ref{NumRes} validates the theoretical findings through numerical simulations using appropriate parameter values and explores the structural sensitivity of the model. Finally, Section \ref{disc} concludes the paper with a discussion of the results.

\section{Mathematical Models}{\label{MathModel}}

The abundance of enriched dynamics prompts us to consider the two-species Bazykin model \cite{Bazykin} for studying structural sensitivity. The model incorporates logistic growth for the prey population, a density-dependent predator death rate, and the Holling type II functional response to represent predation. As an alternative parametrization of the functional response, we consider the Ivlev functional response \cite{Ivlev}, an ecologically well-established prey-dependent function with characteristics similar to the Holling type II response. In general, we denote either functional response as $f(u)$. The key properties common to both functional responses are: \[(i) f(0) = 0;\,\,\, (ii) f'(u) > 0~\forall u \geq 0;\,\,\, (iii) \lim_{u\rightarrow\infty}f(u) =  f^{\infty} < \infty;\,\,\, (iv) f''(u) < 0~\forall u > 0.\] 
Given these shared characteristics, we adopt a generalized, smooth prey-dependent function satisfying the above properties as the functional response and formulate the continuous-time model with Bazykin type reaction kinetics as follows:
 \begin{subequations}\label{algeb_equn}
\begin{eqnarray}
\frac{d u}{d t}&=&ru\left(1-\frac{u}{X_0}\right)-f\left(u\right)v \equiv F_1(u,v),\\
\frac{d v}{d t}&=&f\left(u\right)v-v-k_3{v}^2 \equiv F_2(u,v),
\end{eqnarray}
\end{subequations}
which is non-dimensionalized as in \cite{Shen}. Here, $u\equiv u(t)$ and $v\equiv v(t)$ represent dimensionless concentrations of prey and predator populations, respectively, at time $t$, assuming a homogeneous spatial distribution. However, in reality, individuals of any species move randomly within their habitat, contradicting the assumption of a homogeneous species distribution. To better understand the impact of different parameterizations of the functional response on species dynamics, we extend our temporal model to include spatial interactions. This is achieved by incorporating symmetric diffusion, which effectively captures the random dispersal of species within their habitat. Let $\Omega\subset \mathbb{R}^n$ be a bounded, connected spatial domain representing the species' habitat, with a smooth boundary $\partial \Omega \subset \mathbb{R}^{n-1}$. Following the non-dimensionalization described in \cite{Shen}, we obtained the spatially-extended dimensionless form of the temporal model:
 \begin{subequations}\label{alg_eqn_pde}
\begin{eqnarray}
\frac{\partial u}{\partial t}&=&\Delta u \,\,+\,\,F_1\left(u,v\right),\\
\frac{\partial v}{\partial t}&=&d\Delta v\,\,+\,\,F_2\left(u,v\right),
\end{eqnarray}
\end{subequations}
where $u\equiv u(\mathbf{x},t)$ and $v\equiv v(\mathbf{x},t)$ represent dimensionless population densities of prey and predator, respectively, at time $t$ and spatial position $\mathbf{x}\in \Omega\subset \mathbb{R}^{n}$. Here, $\Delta=\sum_{j=1}^n \frac{\partial^2}{\partial x_j^2}$ denotes the Laplacian operator and $d$ represents the ratio of predator-to-prey self-diffusion coefficients. For our analysis, we assume a two-dimensional spatial domain, although most analytical results remain valid for any $n\in \mathbb{N}$.

The study of system (\ref{algeb_equn}) is conducted under non-negative initial conditions $u(0)\geq 0\mbox{ and }v(0)\geq 0$, while system (\ref{alg_eqn_pde}) is analyzed with non-negative initial conditions, $u(\mathbf{x},0)\geq 0\mbox{ and }v(\mathbf{x},0)\geq 0$ for $\mathbf{x} \in \Omega$ along with no-flux boundary conditions $\frac{\partial u(\mathbf{x},t)}{\partial n}=\frac{\partial v(\mathbf{x},t)}{\partial n}=0$ for $\mathbf{x}\in \partial\Omega$, where $n$ is the outward unit normal to $\partial\Omega$. We examine the system dynamics theoretically and numerically for two distinct functional responses, namely the Holling Type II functional response, denoted as $f_H(u)$, and the Ivlev functional response, denoted as $f_I(u)$, with a suitable choice of parameter values that ensure their close proximity \cite{Fussmann}. To generalize our analytical results beyond these specific cases, we present them for a smooth, generalized functional response $f(u)$ that satisfies the aforementioned properties.

\section{Dynamics of Temporal Model}{\label{LocBif}}

The equilibrium points of system (\ref{algeb_equn}) correspond to the spatially homogeneous steady states of the spatiotemporal system (\ref{alg_eqn_pde}). The analysis of these equilibrium points establishes a baseline understanding of the behaviour of the system and serves as a reference point for more detailed spatial analysis, making their study essential for further progress in spatiotemporal studies. A detailed analysis of the number and stability of equilibrium points for the model with a Holling type II functional response is available in \cite{LU202199}. However, we present these results in the context of a generalized functional response.

\subsection{Equilibria}

The system always has two predator-free equilibrium points: $E_0=(0,0)$ and $E_1=(X_0,0)$, regardless specific parameterization of $f(u)$. However, the potential number of coexisting equilibrium points, $E_*=(u_*,v_*)$, can vary between zero and three, depending on whether the functional response satisfies certain conditions. A coexisting equilibrium corresponds to the intersection of the following two nontrivial nullclines: \[r\left(1-\frac{u}{X_0}\right)=\frac{f(u)}{u}v ~~\mbox{and}~~f(u)-1=k_3v\] within the interior of first quadrant. From the second nullcline, it follows that $v_*=\frac{(f(u_*)-1)}{k_3}$, which is feasible if $f(u_*)>1$, where $u_*$ is a positive real root of the implicit equation 
  \[F(u)\equiv ru\left(1-\frac{u}{X_0}\right)-\frac{\left(\left(f(u)\right)^2-f(u)\right)}{k_3}=0.\]
Thus, determining the number of coexisting equilibrium points requires finding the number of positive real roots of $F(u)=0$ while ensuring $f(u_*)>1$. For notational convenience, we define:
\[G(u)=\frac{dF(u)}{du}=r\left(1-\frac{2u}{X_0}\right)-\frac{f'(u)\left(2f(u)-1\right)}{k_3},\]
\[H(u)=\frac{d^2F(u)}{du^2}=-\frac{2r}{X_0}-\frac{f''(u)(2f(u)-1)+2(f'(u))^2}{k_3}.\]
We first explore all possible scenarios and determine the potential number of coexisting equilibrium points in each case, as outlined in Theorem \ref{thm1}.

\begin{thm}\label{thm1}
 The number of feasible coexisting equilibrium points  in the system \eqref{algeb_equn}, for a fixed parameter value $r$, depends on the  parametrization and the choice of parameter values of the functional response, as follows:

\noindent(a) No coexisting equilibrium point for \[\textbf{C1:}~ f(X_0)<1.\]
\noindent(b) Unique coexisting equilibrium point for \[\textbf{C2:}~f\left(X_0/2\right)\leq 1 < f(X_0).\]
\noindent(c) At least one coexisting equilibrium point for \[\textbf{C3:}~1< f(X_0/2).\] Additionally, at most three coexisting equilibrium points can exist under the condition \begin{equation}\label{ineq1}
\mbox{\textbf{C4:}}~ \frac{f''(f^{-1}(1))+2(f'(f^{-1}(1)))^2}{k_3}<-\frac{2r}{X_0}<\frac{f''\left(\frac{X_0}{2}\right)(2f\left(\frac{X_0}{2}\right)-1)+2\left(f'\left(\frac{X_0}{2}\right)\right)^2}{k_3}.
\end{equation}
\end{thm}

\begin{proof}

(a) The expression of $F(u)$ indicates that $F(u)>0$ for all $u \in [0,X_0]$ if the condition \textbf{\textit{C1}} holds. Consequently, the system \eqref{algeb_equn} has no positive solution for  $u\in [0,X_0]$, ensuring the non-existence of a feasible coexisting equilibrium point.

(b) The expression of $G(u)$ remains negative whenever $u\in [f^{-1}(1), X_0]$. Hence, the function $F(u)$ strictly monotonically decreases over the interval. Additionally, since $F(u)>0$ for $0< u\leq f^{-1}(1)$ and $F(X_0)<0$ under the condition $f\left(X_0/2\right)\leq 1 < f(X_0)$, the intermediate value theorem guarantees the existence of at most one positive solution $u_*\in (f^{-1}(1),X_0)$. Moreover, this solution satisfies the feasibility condition for $v_*$, ensuring the existence of a unique coexisting equilibrium point under this parameter restriction.

(c) The solution of $F(u)=0$ corresponds to the prey density of the coexisting equilibrium points of the system \eqref{algeb_equn}, and this equation can be geometrically represented as the intersection of the curves $v=F(u)$ and $v=0$. Under condition \textbf{\textit{C4}}, we deduce that $H(f^{-1}(1))>0$ and $H(X_0/2)<0$. Since $H(u)$ is a continuous, strictly decreasing function, the intermediate value theorem guarantees the existence of a unique point $u_{CP}\in (f^{-1}(1),X_0/2)$ such that $H(u_{CP})=0$, establishing $u_{CP}$ as the point of inflection of the curve $v=F(u)$.

Since, $G(X_0/2)<0$ and $G(f^{-1}(1))<0$, the slope of the curve $v=F(u)$ is negative at $u=X_0/2$ and $u=f^{-1}(1)$. At the point of inflection of the curve $v=F(u)$,  we must have $G(u_{CP})\geq 0$. Given the continuity of $G(u)$, different scenarios can arise depending on the choice of parameterization, which we discuss next.

\noindent\textbf{\textit{Case1:}} If $G(u_{CP})>0$, then intermediate value theorem assures that $G(u)$ will attain the value $0$ once in the interval $(f^{-1}(1),u_{CP})$, where $H(u)>0$, and once again in the interval $(u_{CP},X_0/2)$ where $H(u)<0$. Consequently, there exists a point $u_m\in (f^{-1}(1),u_{CP})$ where the curve $v=F(u)$ has a minimum, and another point $u_M\in (u_{CP},X_0/2)$ where it attains a maximum. The precise locations of these extrema depend on the choice of parameter values.

Next, we examine different scenarios concerning the shape of the curve $v=F(u)$, which varies based on parameter values that may depend on or be independent of the functional response. Additionally, we analyze the potential number of equilibrium points in each scenario.

\noindent(i) If $F(u_m)>0$ and $F(u_M)>0$, then the system will have only one feasible coexisting equilibrium point. Denote it by $E_*=(u_*,v_*)$, where the prey density of this equilibrium point, $u_*$, satisfies the condition $G(u_*)<0$.

    \noindent(ii) If $F(u_m)=0$ and $F(u_M)>0$, then the system will have two feasible coexisting equilibrium points, with one coexisting equilibrium point having prey density $u_m$. Note that $G(u_m)=0$. The prey density of the other equilibrium point satisfies the condition $G(u_{*})<0$.
    
\noindent(iii) If $F(u_m)<0$ and $F(u_M)>0$, then the system will have three feasible coexisting equilibrium points. Let us denote the equilibrium point with the lowest prey density as $E_{1*}=(u_{1*},v_{1*})$. For this equilibrium, the prey density satisfies the inequality $u_{*}<u_m$, implying the condition $G(u_{1*})<0$. The equilibrium point with medium prey density, $E_{2*}=(u_{2*},v_{2*})$, corresponds to a prey density that satisfies the inequality $u_m<u_{2*}<u_M$. For this equilibrium point, the condition $G(u_{2*})>0$ holds. Finally, the equilibrium point with the highest prey density is denoted by $E_{3*}=(u_{3*},v_{3*})$, where the prey density satisfies the inequality $u_M<u_{3*}$, and we can conclude that the condition $G(u_{3*})<0$ for this equilibrium point.

\noindent(iv) If $F(u_m)<0$ and  $F(u_M)=0$, then the system will have two feasible coexisting equilibrium points, with one having prey density $u_M$. Note that, the prey density of this equilibrium point satisfies the condition $G(u_M)=0$, while the prey density of the other equilibrium point satisfies $G(u_{*})<0$.

\noindent(v) If $F(u_m)<0$ and $F(u_M)<0$, then the system will have only one feasible coexisting equilibrium point. Denote it by $E_*=(u_*,v_*)$, and note that, the prey density satisfies the condition $G(u_*)<0$.

t is important to note that the prey density of any of the coexisting equilibrium points mentioned above must lie within the interval $(f^{-1}(1), X_0)$, meaning the prey density of any coexisting equilibrium point must satisfy the inequality $f(u_*)>1$.

\noindent\textbf{\textit{Case2:}} If $G(u_{CP})=0$, then for all $u\in (f^{-1}(1),X_0)$, we have $G(u)\leq 0$. Consequently, the curve $v=F(u)$ remains monotone increasing throughout this interval, indicating that it can intersect the line $v=0$ at most once. This implies that the system can have, at most, one coexisting equilibrium point.

If, for certain parameter values, the condition \textbf{\textit{C4}} is not satisfied, then the curve $v=F(u)$ does not have a point of inflection. Consequently, there is no possibility of a change in the slope of the curve, and we additionally obtain $G(u)< 0,~\forall~u\in[f^{-1}(1),\frac{X_0}{2}]$. Hence, the curve $v=F(u)$ intersects the line $v=0$ at most once, implying that the system has at most one feasible equilibrium point.\end{proof}

\noindent In the next subsection, we will discuss the local stability of all equilibrium points that the system can exhibit under different scenarios.

\subsection{Local stability of equilibria}\label{temstability}

Species interactions play a crucial role in spatial pattern formation, as diffusive instability alone cannot generate patterns without ecological interactions. Therefore, we focus on coexisting equilibrium points, where both species stably persist. To examine changes in the stability of these equilibrium points, we analyze their stability using linear stability analysis. This approach investigates the stable and unstable eigenspaces of the Jacobian matrix, whose general form for the system \eqref{algeb_equn} evaluated at an arbitrary point is given by:
\begin{eqnarray*}
    J(E) = \left[
    \begin{array}{ccc}
    r\left(1-\frac{2u}{X_0}\right)-f'(u)v & -f(u) \\
    f'(u)v & f(u)-1-2k_3v
    \end{array}\right].
\end{eqnarray*} 
The Jacobian matrix, evaluated at a typical coexisting equilibrium point $E_*=(u_*,v_*)\equiv(u_*,(f(u_*)-1)/k_3)$, is given by:
\begin{eqnarray*}
    J(E_*) = \left[
    \begin{array}{ccc}
    r\left(1-\frac{2u_*}{X_0}\right)-\frac{f'(u_*)\left(f(u_*)-1\right)}{k_3} & -f(u_*) \\
    \frac{f'(u_*)\left(f(u_*)-1\right)}{k_3} & 1-f(u_*)
    \end{array}\right].
\end{eqnarray*}
Using the Routh-Hurwitz criteria \cite{routh}, we can conclude that a coexisting equilibrium point is asymptotically stable if the following two conditions are satisfied:\[\text{Trace}(J(E_*))=r\left(1-\frac{2u_*}{X_0}\right)-\frac{f'(u_*)\left(f(u_*)-1\right)}{k_3}+(1-f(u_*))<0,\] \[\mbox{and}~\text{Det}(J(E_*))=\left(1-f(u_*)\right)G(u_*)>0.\]
Based on the different scenarios outlined in Theorem \ref{thm1}, we now discuss the stability of various coexisting equilibrium points. Whenever the system exhibits exactly one coexisting equilibrium point, $E_*=(u_*,v_*)$, the prey density at this equilibrium satisfies the condition $G(u_*)\leq 0$. Consequently, it is straightforward to verify that $\text{Det}(J(E_*))\geq0$. However, if the prey density $u_*$ of the coexisting equilibrium point satisfies the condition
 \begin{equation}\label{trace}
r\left(1-\frac{2u}{X_0}\right)<\left(f(u)-1\right)\left(\frac{f'(u)}{k_3}+1\right),
\end{equation}
then, the equilibrium point is stable if $\text{Trace}(J(E_*))<0$; otherwise, it remains unstable.

Whenever the system exhibits three coexisting equilibrium points, the equilibrium point $E_{2*}$ is inherently unstable whenever it exists, as we have already established the condition $G(u_{2*})>0$, which implies $\text{Det}(J(E_{2*}))<0$. However, for the equilibrium points $E_{1*}\mbox{~and~}E_{3*}$, we have $\text{Det}(J(E_{j*}))>0~\mbox{for}~j=1,2$. Therefore, the stability of these two equilibrium points hinges on whether the prey densities satisfy the condition \ref{trace}. If they do, the equilibrium point is stable; otherwise, it is unstable.

In the following subsection, we discuss the various potential bifurcations that the system might undergo near the coexisting equilibrium points. Additionally, we explore the specific conditions that could lead to these bifurcations. 

\subsection{Local bifurcation results}
In the context of studying structural sensitivity, we consider the generalized system (\ref{algeb_equn}) with a generalized functional response $f(u)$ and examine the number of coexisting equilibrium points and their stability under different parameter variations in $r,X_0~\mbox{and}~k_3$. Additionally, we use Sotomayor's theorem \cite{perko2013} to verify the transversality conditions for the local bifurcations obtained. To simplify our subsequent discussion, let us denote:\[F\left(u,v\right) \equiv \left[\begin{array}{c}
ru\left(1-\frac{u}{X_0}\right)-f\left(u\right)v  \\
f\left(u\right)v-v-k_3{v}^2
\end{array}\right].\]\\
First, we determine the conditions for the occurrence of a saddle-node bifurcation, analytically derive its threshold value, and verify the transversality conditions.

\subsubsection{Saddle-Node bifurcation}
From the discussion of Theorem \ref{thm1}, we know that when either $u_m$ or $u_M$ is a root of $F(u)=0$, the system has exactly two coexisting equilibrium points, with at least one having a prey density of $u_M$ or $u_m$. We refer to these particular nontrivial equilibrium points as $E_{SN_1}\equiv\left(u_M,v_M\right)$ and $E_{SN_2}\equiv\left(u_m,v_m\right)$, where $v_m=\frac{f(u_m)-1}{k_3}~\mbox{and}~v_M=\frac{f(u_M)-1}{k_3}$.Without loss of generality, we denote either equilibrium point as $E_{SN}=(u_{SN},v_{SN})$.  The Jacobian matrix evaluated at $E_{SN}$ is given by: 
\begin{eqnarray*}
    J(E_{SN}) = \left[
    \begin{array}{ccc}
    \frac{f'(u_{SN})f(u_{SN})}{k_{3}} & -f(u_{SN}) \\
    \frac{f'(u_{SN})\left(f(u_{SN})-1\right)}{k_{3}} & 1-f(u_{SN})
    \end{array}\right].
\end{eqnarray*}
A few steps of row operations show that the Jacobian matrix has rank one, indicating that zero is a simple eigenvalue of $J(E_{SN})$. Therefore, $E_{SN}$ is non-hyperbolic. The eigenvectors corresponding to this zero eigenvalue of the matrices $J(E_{SN})$ and $\left[J(E_{SN})\right]^T$ are given by
$$V = \left[\begin{array}{c}
1 \\
\frac{f'(u_{SN})}{k_{3}} \\
\end{array}\right] ~\mbox{and}~ W = \left[\begin{array}{c}
\frac{\left(1-f(u_{SN})\right)}{f(u_{SN})} \\
1 \\
\end{array}\right],$$ respectively. Here $\left[J(E_{SN})\right]^{T}$ denotes the transpose of the matrix $J(E_{SN})$. Theorem \ref{thm1} assures that $f(u_{SN})>1$, and whenever $u_{SN}\neq u_{CP}$, we have $H(E_{SN})\neq 0$. Using this informations, we obtain the transversality conditions following \cite{perko2013} as follows:
\[W^T\left[\frac{\partial F(u_{SN},\frac{f(u_{SN})-1}{k_{3}})}{\partial k_3} \right]= -\left(\frac{f(u_{SN})-1}{k_{3}}\right)^2\neq  0,~\mbox{and}\]
\[W^T\left[D^2F\left(u_{SN},\frac{f(u_{SN})-1}{k_{3}}\right)(V,V)\right] =\left(\frac{f(u_{SN})-1}{f(u_{SN})}\right)H(u_{SN}) \neq  0.\]
This implies that the system undergoes non-degenerate saddle-node bifurcations. Here, we have considered $k_3$ as the bifurcation parameter. However, other parameters can also yield these bifurcation thresholds, which are not covered here.

Note that the system experiences two saddle-node bifurcations of coexisting equilibrium points under certain choices of parameter values. In the next subsection, we will discuss a scenario where these two saddle-node bifurcations coincide.

\subsubsection{Cusp bifurcation}
With an appropriate choice of parameter values, it is possible to obtain a scenario in which the point $u_{CP}$ becomes the root of both the algebraic equations $F(u)=0$ and $G(u)=0$. In that scenario, the system exhibits a unique coexisting equilibrium point with prey density $u_{CP}$. Let us denote this unique equilibrium point as $E_{CP}\equiv\left(u_{CP},v_{CP}\right)$, where $v_{CP}=\frac{f(u_{CP})-1}{k_3}$. Proceeding as above, if we evaluate the Jacobian matrix $J(E_{CP})$, a simple and straightforward calculation shows that the matrix $J(E_{CP})$ has rank one and that the zero eigenvalue of $J(E_{CP})$ is simple, assuring that the equilibrium point is non-hyperbolic. The eigenvectors corresponding to this zero eigenvalue of the matrix $J(E_{CP})$ and $\left[J(E_{CP})\right]^T$ are
$$V = \left[\begin{array}{c}
1 \\
\frac{f'(u_{CP})}{k_{3}} \\
\end{array}\right] ~\mbox{and}~ W = \left[\begin{array}{c}
\frac{\left(1-f(u_{CP})\right)}{f(u_{CP})} \\
1 \\
\end{array}\right]$$ respectively. Using these eigenvectors, we can obtain the quadratic normal form coefficient evaluated at $E_{CP}$ as follows: 
\[W^T\left[D^2F\left(u_{CP},\frac{f(u_{CP})-1}{k_{3}}\right)(V,V)\right] =\left(\frac{f(u_{CP})-1}{f(u_{CP})}\right)H(u_{CP}) = 0,\]
since $H(u_{CP}) = 0$. From here, it can be concluded that the system undergoes a cusp bifurcation under this specific choice of parameter values.

From Theorem \ref{thm1}, it is evident that the system may exhibit at most three coexisting equilibrium points. In the next subsection, we discuss the results concerning the change in stability of these coexisting equilibrium points.

\subsubsection{Hopf bifurcation}
The system exhibits unique coexisting equilibrium point under certain scenarios, and the prey density of that unique equilibrium point must satisfy the inequalities $f(u_*)>1$ and $G(u_*)\leq 0$,  as discussed in Theorem \ref{thm1}. The equality $G(u_*)= 0$ occurs only when $u_*=u_{CP}$; otherwise $G(u_*)<0$. Consequently, we can conclude that $\text{Det}(J(E_*))>0$, whenever $u_*\neq u_{CP}$. However, the stability behavior of that unique coexisting equilibrium point changes when the condition $\text{Trace}(J(E_*))=0$ holds for some specific choice of parameter(s) value.

Theorem \ref{thm1} also addresses the scenario in which the system exhibits three coexisting equilibrium points. The equilibrium point with medium prey density cannot undergo a Hopf bifurcation because $G(u_{2*})>0$, which further implies $\text{Det}(J(E_{2*}))<0$. In contrast, for the equilibrium points associated with higher and lower prey densities, we have $G(u_{1*})<0$ and $G(u_{3*})<0$, ensuring that $\text{Det}(J(E_{1*}))>0$ and $\text{Det}(J(E_{3*}))>0$ respectively. If either of these two equilibrium points additionally satisfies the condition $\text{Trace}(J(E_*))=0$ for a specific choice of parameter(s) value, the stability of that equilibrium point undergoes an alteration.

Since we have previously chosen $k_3$ as the bifurcation parameter to obtain saddle-node bifurcations while keeping other parameters fixed, we will also consider $k_3$ in this context. Based on the scenarios discussed above, there exists a threshold value $k_{3H}$ for fixed values of the other parameters such that the stability of equilibrium points undergoes alteration. It can be shown that the system satisfies the following transversality conditions at the parametric threshold $k_3=k_{3H}$:
 \[{\text{Det}(J(E_*))}>0~~\mbox{and}\] \[\frac{d}{dk_3}{\left(\text{Trace}(J(E_*))\right)}_{|k_3=k_{3H}}=\frac{f'(u_*)\left(f(u_*)-1\right)}{k_{3H}^2}\neq 0.\]
Here $E_*=(u_*,v_*)$ represents the coexisting equilibrium point that experiences a stability alteration at this threshold value. This transversality condition implies that the coexisting equilibrium point undergoes a Hopf bifurcation under certain parametric choices.

Notably, if $\text{Trace}(J(E_{3*}))=0$ holds for a particular set of parameter values, then we must have $\text{Trace}(J(E_{1*}))>0$ for the same parameter choice. Conversely, if $\text{Trace}(J(E_{1*}))=0$ holds for some specific parameter values, then, under that parameter restriction, we will have $\text{Trace}(J(E_{3*}))<0$. This implies that coexisting equilibrium points with higher and lower prey densities cannot undergo Hopf bifurcations simultaneously.

In this subsection, we have only considered cases where the system exhibits either a unique or three coexisting equilibrium points. However, the system can also exhibit two coexisting equilibrium points. In the next subsection, we will discuss the scenario of a Hopf bifurcation when the system exhibits exactly two coexisting equilibrium points.

\subsubsection{Bogdanov-Takens bifurcation}
From the discussion in the subsection on saddle-node bifurcation, we know that whenever the system exhibits two coexisting equilibrium points, one of them must be the equilibrium point $E_{SN}=(u_{SN},v_{SN})$. Depending on the choice of parameter values, the equilibrium point $E_{SN}$ undergoes a Hopf bifurcation. We denote the equilibrium point at this bifurcation threshold as $E_{BT}=(u_{BT},v_{BT})$. At this parameter threshold, $E_{BT}$ satisfies the following conditions: \[\text{Trace}(J(E_{BT}))=0~\mbox{and}~\text{Det}(J(E_{BT}))=0.\] This bifurcation is known as the Bogdanov-Takens bifurcation. The transversality conditions for the Bogdanov-Takens bifurcation are also satisfied when evaluated at the equilibrium point $E_{BT}$, and they are given as follows:\begin{multline*}
    \frac{\partial^2F_1}{\partial u^2}-\frac{\frac{\partial F_1}{\partial u}\frac{\partial^2 F_1}{\partial u\partial v}}{\frac{\partial F_1}{\partial v}}+\frac{\partial^2F_2}{\partial u \partial v}
=\\-\frac{2r}{X_0}-f''(u_{BT})v_{BT}-\frac{\left(r-\frac{2ru_{BT}}{X_0}-f'(u_{BT})v_{BT}\right)f'(u_{BT})}{f(u_{BT})}+f'(u_{BT})\neq 0,\end{multline*}
\begin{multline*}\frac{1}{2}\frac{\partial F_1}{\partial u}\frac{\partial^2F_1}{\partial u^2}-\frac{\left(\frac{\partial F_1}{\partial u}\right)^2\frac{\partial^2 F_1}{\partial u\partial v}}{\frac{\partial F_1}{\partial v}}+\frac{1}{2}\frac{\partial F_1}{\partial v}\frac{\partial^2F_1}{\partial u^2}-\frac{\partial F_1}{\partial u}\frac{\partial^2F_2}{\partial u \partial v}=\\\left(-\frac{r}{X_0}-\frac{f''(u_{BT})v_{BT}}{2}-f'(u_{BT})\right)\left(r-\frac{2ru_{BT}}{X_0}-f'(u_{BT})v_{BT}\right)+f(u_{BT})k_3\\-\left(r-\frac{2ru_{BT}}{X_0}-f'(u_{BT})v_{BT}\right)\left(r-\frac{2ru_{BT}}{X_0}-f'(u_{BT})v_{BT}\right)\frac{f'(u_{BT})}{f(u_{BT})}\neq 0.\end{multline*}

\noindent These transversality conditions helps us to put the system into the following normal form: 
\begin{subequations}
\begin{eqnarray}
\frac{d u}{d t}&=&v,\\
\frac{d v}{d t}&=& \zeta_1(\varepsilon_1,\varepsilon_2)+\zeta_2(\varepsilon_1,\varepsilon_2)v+u^2+\frac{d_2}{\sqrt{d_1}}uv,
\end{eqnarray}
\end{subequations}
where $d_1=\frac{1}{2}\frac{\partial F_1}{\partial u}\frac{\partial^2F_1}{\partial u^2}-\frac{\left(\frac{\partial F_1}{\partial u}\right)^2\frac{\partial^2 F_1}{\partial u\partial v}}{\frac{\partial F_1}{\partial v}}+\frac{1}{2}\frac{\partial F_1}{\partial v}\frac{\partial^2F_1}{\partial u^2}-\frac{\partial F_1}{\partial u}\frac{\partial^2F_2}{\partial u \partial v},~d_2=\frac{\partial^2F_1}{\partial u^2}-\frac{\frac{\partial F_1}{\partial u}\frac{\partial^2 F_1}{\partial u\partial v}}{\frac{\partial F_1}{\partial v}}+\frac{\partial^2F_2}{\partial u \partial v}$ and $\varepsilon_1,\varepsilon_2$ are very small perturbation ($|\varepsilon_1|,|\varepsilon_2|\ll 1$) of the parameter values from the bifurcation threshold.

It is important to note that all the previously discussed temporal bifurcations occur in the vicinity of equilibrium points, where the species population is assumed to be homogeneously distributed in space. However, from an ecological perspective, studying homogeneous equilibrium points alone is insufficient for making accurate predictions about the long-term behavior of species. To improve the accuracy of predictions regarding the future dynamics of species habitats, it is necessary to incorporate spatial variables into the analysis.

 \section{Spatio-Temporal Model: Analytical results}{\label{SpatioModel}}
 The diffusion-driven continuous-time system under consideration, which incorporates Bazykin reaction kinetics with a generalized functional response, is given by: 
\begin{subequations}\label{alg_eqn_pde9}
\begin{eqnarray}
\frac{\partial u}{\partial t}&=&\Delta u \,\,+\,\,ru\left(1-\frac{u}{X_0}\right)-f\left(u\right)v,\\
\frac{\partial v}{\partial t}&=&d\Delta v\,\,+\,\,f\left(u\right)v-v-k_3{v}^2,
\end{eqnarray}
\end{subequations}
which has been non-dimensionalized as in \cite{Shen}, subject to the non-negative initial conditions $u(\mathbf{x},0)\geq 0\mbox{ and }v(\mathbf{x},0)\geq 0$ for $\mathbf{x} \in \Omega$ along with no-flux boundary conditions $\frac{\partial u(\mathbf{x},t)}{\partial n}=\frac{\partial v(\mathbf{x},t)}{\partial n}=0$ for $\mathbf{x}\in \partial\Omega$, where $n$ is the outward unit normal to $\partial\Omega$. We establish the existence of its global solution along with a priori bounds under specified initial and boundary conditions in Section \ref{Existence}. We analyze the stability of homogeneous steady states in Section \ref{localstablespace} and derive the criteria for the existence and non-existence of heterogeneous steady states in Section \ref{HetSteb}.

 \subsection{Existence and Boundedness of global solution}\label{Existence}
 
\begin{thm}\label{thm2} Let us denote $\mathcal{D}=\Omega\times (0,\infty)$, $\Bar{\mathcal{D}}=\Bar{\Omega}\times (0,\infty)$. If we choose the parameter values related to functional response in such a way that $f(X_0) >1$, then for the system (\ref{alg_eqn_pde9}),

\noindent(a) if $u_{0}(\mathbf{x},0) \geq 0$ and $v_0(\mathbf{x},0)\geq 0$, we can find a global solution $(u(\mathbf{x},t),v(\mathbf{x},t))$. If our initial condition is identically zero, then $(u(\mathbf{x},t),v(\mathbf{x},t))\equiv (0,0)$.% If $f(sup_{\Bar{\Omega}})u_{CP}(x)<1$ then $v(x,t)\rightarrow 0 $ uniformly as $t \rightarrow \infty$

\noindent(b) If $k_3> \frac{f^{\infty}(f(X_0)-1)}{r}$ and $f\left(X_0-\left(\left[\frac{X_0f^{\infty}(f(X_0)-1)}{rk_3}\right]\right)\right)>1$, then any global solution $(u(\mathbf{x},t),v(\mathbf{x},t))$ of the system (\ref{alg_eqn_pde9}) satisfies the inequalities \[\underset{t \rightarrow \infty}{\limsup} \max_{\mathbf{x} \in \Bar{\Omega}}u(\mathbf{x},t)\leq X_0\mbox{~;~}\underset{t \rightarrow \infty}{\limsup} \max_{\mathbf{x} \in \Bar{\Omega}}v(\mathbf{x},t)\leq \left[(f(X_0)-1)/k_3\right],\] and \[\underset{t \rightarrow \infty}{\liminf} \min_{\mathbf{x} \in \Bar{\Omega}}u(\mathbf{x},t)\geq X_0-\left(\left[\frac{X_0f^{\infty}(f(X_0)-1)}{rk_3}\right]\right);\]
\[\underset{t \rightarrow \infty}{\liminf} \min_{\mathbf{x} \in \Bar{\Omega}}v(\mathbf{x},t)\geq \left[f\left(X_0-\left(\left[\frac{X_0f^{\infty}(f(X_0)-1)}{rk_3}\right]\right)\right)-1\right]/k_3.\]
i.e., the positive non-constant steady state persists.

\noindent(c) For any global solution $(u(\mathbf{x},t),v(\mathbf{x},t))$ of the system (\ref{alg_eqn_pde9}),
$$\underset{t \rightarrow \infty}{\limsup} \int_{\Omega} u(\mathbf{x},t) \,d\mathbf{x} \leq \frac{(r+1)^2X_0}{4r}|\Omega| \,\,;\,\, \underset{t \rightarrow \infty}{\limsup} \int_{\Omega} v(\mathbf{x},t) \,d\mathbf{x} \leq \frac{(f\left(X_0
) -1\right)}{k_3}|\Omega|.$$
\end{thm}

\begin{proof} (a) For $(u,v) \in \mathbb{R}^2_{\geq 0}=\{(u,v)\in \mathbb{R}^2: u \geq 0, v\geq 0\}$, we have the conditions $\frac{\partial F_1(u,v)}{\partial v}\leq 0$ and $\frac{\partial F_2(u,v)}{\partial u}\geq 0$. Hence the system (\ref{alg_eqn_pde9}) is a mixed quasimonotone system \cite{pao}. Let $\hat{u}(t)$ be the solution of the following initial value problem
\begin{equation}
   \begin{cases}
     \frac{du}{dt}=ru\left(1-\frac{u}{X_0}\right),\\
     u(0)=\sup_{\mathbf{x} \in \Bar{\Omega}}u(\mathbf{x},0),
   \end{cases}\,
\end{equation}
and $\hat{v}(t)$ be the solution of the following initial value problem
\begin{equation}
   \begin{cases}
     \frac{dv}{dt}=\left(f(\hat{u}(t)) -1\right)v-k_3v^2,\\
     v(0)=\sup_{\mathbf{x}
     \in \Bar{\Omega}}v(\mathbf{x},0).
   \end{cases}\,
\end{equation}
 Clearly, depending upon the supremum of the initial conditions $u(\mathbf{x},0)\mbox{ and }v(\mathbf{x},0)$, we can state that $\hat{u}(t) < \max\{u(0),X_0\}$ and $\hat{v}(t)< \max\{v(0), \frac{f(X_0)-1}{k_3}\}$ for all $t \in (0, \infty)$ and $\underset{t \rightarrow \infty}{lim}\hat{u}(t)=X_0\mbox{; }\underset{t \rightarrow \infty}{lim}\hat{v}(t)=\frac{f(X_0)-1}{k_3}$. Now if we consider $\left(\overline{u}(\mathbf{x},t),\overline{v}(\mathbf{x},t)\right)=\left(\hat{u}(t),\hat{v}(t)\right)$ and $(\underline{u}(\mathbf{x},t),\underline{v}(\mathbf{x},t))=(0,0)$, we can establish the following inequalities: \[\frac{\partial \overline{u}(\mathbf{x},t)}{\partial t}-\Delta \overline{u}(\mathbf{x},t)-F_1\left(\overline{u}(\mathbf{x},t),\underline{v}(\mathbf{x},t)\right)=0\geq0=\frac{\partial \underline{u}(\mathbf{x},t)}{\partial t}-\Delta \underline{u}(\mathbf{x},t)-F_1\left(\underline{u}(\mathbf{x},t),\overline{v}(\mathbf{x},t)\right),\]
\[\,\,\frac{\partial \overline{v}(\mathbf{x},t)}{\partial t}-d\Delta \overline{v}(\mathbf{x},t)-F_2\left(\overline{u}(\mathbf{x},t),\overline{v}(\mathbf{x},t)\right)=0\geq 0=\frac{\partial \underline{v}(\mathbf{x},t)}{\partial t}-d\Delta \underline{v}(\mathbf{x},t)-F_2\left(\underline{u}(\mathbf{x},t),\underline{v}(\mathbf{x},t)\right),\]
together with the inequality of initial conditions: \[\overline{u}(\mathbf{x},0)=\sup_{\mathbf{x} \in \Bar{\Omega}}u(\mathbf{x},0)\geq 0=\underline{u}(\mathbf{x},0)\mbox{; }\overline{v}(\mathbf{x},0)=\sup_{\mathbf{x} \in \Bar{\Omega}}v(\mathbf{x},0)\geq 0=\underline{v}(\mathbf{x},0),\]
and the inequality of boundary conditions:
\[\frac{\partial \overline{u}(\mathbf{x},t)}{\partial n}=0\geq 0=\frac{\partial \underline{u}(\mathbf{x},t)}{\partial n}\mbox{ ; }\frac{\partial \overline{v}(\mathbf{x},t)}{\partial n}=0\geq 0=\frac{\partial \underline{v}(\mathbf{x},t)}{\partial n}.\]
From all of these above inequalities, we can conclude that $\left(\overline{u}(\mathbf{x},t),\overline{v}(\mathbf{x},t)\right)$ and $(\underline{u}(\mathbf{x},t),\underline{v}(\mathbf{x},t))$ are the coupled upper and lower solution of the mixed quasimonotone system (\ref{alg_eqn_pde9}), respectively \cite{pao}. Now using the existence-comparison theorem $\left(\mbox{Theorem }3.3\mbox{ of chapter }8\mbox{ \cite{pao}}\right)$, we can conclude that the system (\ref{alg_eqn_pde9}) has a unique global solution $\left(u(\mathbf{x},t),v(\mathbf{x},t)\right)$ such that 
\[0\leq u(\mathbf{x},t) \leq \hat{u}(t) \mbox{ and }0\leq v(\mathbf{x},t) \leq \hat{v}(t).\]
Additionally, if we choose the initial conditions $\left(u(\mathbf{x},0),v,(\mathbf{x},0)\right)\equiv (0,0)$, then we can conclude that $\left(u(\mathbf{x},t),v(\mathbf{x},t)\right)\equiv (0,0)$.

\noindent(b) Let us choose an $\epsilon>0$, sufficiently small, in such a way that $r-f^{\infty}\left(\left[\frac{(f(X_0)-1)}{k_3}\right]+\epsilon\right)>0$. If the initial conditions are not identically zero, then using the limit: $\underset{t \rightarrow \infty}{\lim}\hat{u}(t)=X_0$ and $\underset{t \rightarrow \infty}{\lim}\hat{v}(t)=(f(X_0)-1)/k_3$, we can find $T_1\in (0,\infty)$ depending upon $\epsilon$ such that $ u(\mathbf{x},t)\leq X_0+\epsilon$ and $ v(\mathbf{x},t)\leq \left[(f(X_0)-1)/k_3\right]+\epsilon,~\forall~ \mathbf{x}\in \Omega\mbox{~whenever~}t\geq T_1$. Using the arbitrariness of $\epsilon$, one can easily conclude that 
\[\underset{t \rightarrow \infty}{\limsup} \max_{\mathbf{x} \in \Bar{\Omega}}u(\mathbf{x},t)\leq X_0\mbox{~;~}\underset{t \rightarrow \infty}{\limsup} \max_{\mathbf{x} \in \Bar{\Omega}}v(\mathbf{x},t)\leq \left[(f(X_0)-1)/k_3\right].\]

Also using the inequality $v(\mathbf{x},t)\leq \left[(f(X_0)-1)/k_3\right]+\epsilon,\mbox{~for~all~} (\mathbf{x},t)\in \Omega \times [T_1, \infty)$, we further obtain the inequality
\[\,\,\,\frac{\partial u(\mathbf{x},t)}{\partial t}-\Delta u(\mathbf{x},t)-ru(\mathbf{x},t)\left(1-\frac{u(\mathbf{x},t)}{X_0}\right)+f^{\infty}u(\mathbf{x},t)v(\mathbf{x},t)\]
\[\geq0=\frac{\partial z(t)}{\partial t}-\Delta z(t)-rz(t)\left(1-\frac{z(t)}{X_0}\right)+f^{\infty}u(\mathbf{x},t)\left(\left[\frac{(f(X_0)-1)}{k_3}\right]+\epsilon\right)\]
together with the inequality of initial condition
\[u(\mathbf{x},T_1)\geq \min_{\mathbf{x} \in \Bar{\Omega}}u(\mathbf{x},T_1)=z(T_1),\]
and boundary condition 
\[\frac{\partial u(\mathbf{x},t)}{\partial n}=0\geq 0=\frac{\partial z(t)}{\partial n},\]
where $z(t)$ is the solution to the initial value problem
\begin{equation}
   \begin{cases}
     \frac{dz}{dt}=rz\left(1-\frac{z}{X_0}\right)-f^{\infty}z\left(\left[\frac{(f(X_0)-1)}{k_3}\right]+\epsilon\right),~ t>T_1,\\
     z(T_1)=\underset{\mathbf{x} \in supp(u(\mathbf{x},T_1))}{\min}u(\mathbf{x},T_1)>\underset{\mathbf{x} \in supp\left(u(\mathbf{x},0)\right)}{\min}u(\mathbf{x},0)\geq 0.
   \end{cases}\,
\end{equation}
Therefore, using the comparison principle \cite{pao}, we can have the inequality $u(\mathbf{x},t)\geq z(t)$, i.e., $u(\mathbf{x},t)$ be an upper solution \cite{pao}. Since $\underset{t \rightarrow \infty}{\lim} z(t)=\frac{X_0}{r}\left[r-f^{\infty}\left(\left[\frac{(f(X_0)-1)}{k_3}\right]+\epsilon\right)\right]$, using this limit and arbitrariness of $\epsilon$, we can conclude that \[\underset{t \rightarrow \infty}{\liminf} \min_{\mathbf{x} \in \Bar{\Omega}}u(\mathbf{x},t)\geq X_0-\left(\left[\frac{X_0f^{\infty}(f(X_0)-1)}{rk_3}\right]\right)\geq 0.\]
Now, depending upon $\epsilon$, we can find $T_2\in(0,\infty)$ such that $\forall \mathbf{x}\in \Omega$, $u(\mathbf{x},t)\geq \left[X_0-\left(\left[\frac{X_0f^{\infty}(f(X_0)-1)}{rk_3}\right]\right)\right]-\epsilon,$ whenever $t\geq T_2$. Using this inequality, we obtain
\[\,\,\,\,\,\,\,\,\,\,\,\,\,\,\,\,\,\frac{\partial v(\mathbf{x},t)}{\partial t}-\Delta v(\mathbf{x},t)-f\left(X_0-\left(\left[\frac{X_0f^{\infty}(f(X_0)-1)}{rk_3}\right]\right)\right)v(\mathbf{x},t)+v(\mathbf{x},t)+k_3v^2(\mathbf{x},t)\]
\[\geq0=\frac{\partial w(t)}{\partial t}-\Delta w(t)-f\left(X_0-\left(\left[\frac{X_0f^{\infty}(f(X_0)-1)}{rk_3}\right]\right)\right) v(\mathbf{x},t)+v(\mathbf{x},t)+k_3v^2(\mathbf{x},t),\]
together with the inequality of initial condition
\[v(\mathbf{x},T_2)\geq \min_{\mathbf{x} \in \Bar{\Omega}}v(\mathbf{x},T_2)=w(T_2),\]
and the boundary condition 
\[\frac{\partial v(\mathbf{x},t)}{\partial n}=0\geq 0=\frac{\partial w(t)}{\partial n},\]
where $w(t)$ is the solution to the initial value problem
\begin{equation}
   \begin{cases}
     \frac{dw}{dt}=w(t)\left[f\left(X_0-\left(\left[\frac{X_0f^{\infty}(f(X_0)-1)}{rk_3}\right]\right)\right)-1-k_3w(t)\right],~ t>T_2,\\
     w(T_2)=\underset{\mathbf{x} \in supp(v(\mathbf{x},T_2))}{\min}v(\mathbf{x},T_2)>\underset{\mathbf{x} \in supp(v(\mathbf{x},0))}{\min}v(\mathbf{x},0)\geq 0.
   \end{cases}\,
\end{equation}
Therefore, using the comparison principle \cite{pao}, we can have the inequality $v(\mathbf{x},t)\geq w(t)$, i.e., $v(\mathbf{x},t)$ becomes an upper solution. Check that $\underset{t \rightarrow \infty}{\lim} w(t)=\left[f\left(X_0-\left(\left[\frac{X_0f^{\infty}(f(X_0)-1)}{rk_3}\right]\right)\right)-1\right]/k_3$, it is straight forward to conclude that \[\underset{t \rightarrow \infty}{\liminf} \min_{\mathbf{x} \in \Bar{\Omega}}v(\mathbf{x},t)\geq \left[f\left(X_0-\left(\left[\frac{X_0f^{\infty}(f(X_0)-1)}{rk_3}\right]\right)\right)-1\right]/k_3.\]

\noindent(c) Let us consider $P(t)=\displaystyle \int_{\Omega} u(x,t)\,dx$ and $Q(t)=\displaystyle \int_{\Omega} v(x,t)\,dx$, Then by using divergence theorem and boundary conditions, we can write 
\begin{equation}
   \begin{cases}
     \frac{dP(t)}{dt}=\displaystyle \int_{\Omega} \left[ru(\mathbf{x},t)\left(1-\frac{u(\mathbf{x},t)}{X_0}\right)-f(u(\mathbf{x},t))v(\mathbf{x},t)\right] d\mathbf{x},\\
     \frac{dQ(t)}{dt}= \displaystyle \int_{\Omega} [f(u(\mathbf{x},t))v(\mathbf{x},t)-v(\mathbf{x},t)-k_3v^2(\mathbf{x},t)] d\mathbf{x}.
   \end{cases}\,
\end{equation}
From the proof of Theorem \ref{thm2} (b), we know that, for all $t> T_1$, $v(\mathbf{x},t)\leq \frac{(f(X_0)-1)}{k_3}+\epsilon$. Hence we can show that $$\underset{t \rightarrow \infty}{\lim \sup} \int_{\Omega} v(\mathbf{x},t) \,d\mathbf{x} \leq \frac{(f(X_0) -1)}{k_3}|\Omega|.$$

Also from the equation $\frac{d(P(t)+Q(t))}{dt}=\int_{\Omega}[ru(1-\frac{u}{X_0})-v-k_3v^2]\, d\mathbf{x}$, we can say $P(t)+Q(t) \leq \frac{(r+1)^2X_0}{4r}|\Omega|$ which gives an upper bound for $P(t)$. Therefore we can clearly say $$\underset{t \rightarrow \infty}{\lim \sup} \int_{\Omega} u(\mathbf{x},t) \,d\mathbf{x} \leq \frac{(r+1)^2X_0}{4r}|\Omega|.$$
\end{proof}

 \subsection{Homogeneous steady-state analysis}\label{localstablespace}
 
From the discussion in Section \ref{LocBif}, we have already gained comprehensive insights into the homogeneous steady states of system (\ref{alg_eqn_pde9}), which include a trivial homogeneous steady state $E_0$, a semi-trivial homogeneous steady state $E_1$, and at most three coexisting homogeneous steady states. In this subsection, we focus on analyzing the stability of these steady states under spatially heterogeneous perturbations. Let us represent a typical homogeneous steady state as $E=(u_{eq},v_{eq})$. Introducing spatial perturbation as $u(\mathbf{x},t)=u_{eq}+\hat{u}(\mathbf{x},t)$ and $v(\mathbf{x},t)=v_{eq}+\hat{v}(\mathbf{x},t)$, and linearizing the system (\ref{alg_eqn_pde9}) about the steady state, we obtain:
\[\left[
    \begin{array}{ccc}
    \frac{\partial \hat{u}}{\partial t} \\
    \frac{\partial \hat{v}}{\partial t}
    \end{array}\right]=\mathcal{L}\left[\begin{array}{ccc}\hat{u} \\\hat{v}
    \end{array}\right]=(J(E)+\mathcal{D}\Delta)\left[\begin{array}{ccc}\hat{u} \\\hat{v}
    \end{array}\right],\]
where $(\hat{u}(\mathbf{x},t),\hat{v}(\mathbf{x},t))\in \{C(\bar{\Omega}\times [0,\infty))\cap C^{2,1}(\bar{\Omega}\times [0,\infty))\}^2$, $\mathcal{D}=diag(1,d)$ and $J(E)$ represents the Jacobian matrix evaluated at homogeneous steady state. Assume that $0=\mu_0<\mu_1<\cdots<\mu_j<\cdots$ are the eigenvalues of the operator $-\Delta$ on $\Omega$ with the homogeneous Neumann boundary condition and $\mathrm{E}_{\mu_j}$ be the eigenfunction space corresponding to the eigenvalue $\mu_j$. Let $\{\phi_{ij}|i=1,2,\cdots, dim(\mathrm{E}_{\mu_j}) \}$ be the orthogonal basis set of $\mathrm{E}_{\mu_j}$. Let us denote $\mathcal{S}_{\Delta}=\{\mu_0,\mu_1,\mu_2,\cdots\}$. If $\mathcal{U}_{ij}=\{c\phi_{ij}: c=\left[c_1,c_2\right]^T; c_1,c_2\in \mathbb{R}\}$, then it can be verify easily that each $\mathcal{U}_{j}$ is invariant under the operator $\mathcal{L}$, where \[\mathcal{U}_{j}=\overset{\dim(\mathrm{E}_{\mu_j})}{\underset{i=1}{\bigoplus}}\mathcal{U}_{ij},\]
and also it can be shown that
\[\mathcal{U}=\Biggl\{(\psi_1,\psi_2)^T\in C^2(\bar{\Omega})\times C^2(\bar{\Omega}): \frac{\partial \psi_1}{\partial n}=\frac{\partial \psi_2}{\partial n}=0\mbox{~for~}x\in \partial\Omega\Biggl\}=\overset{\infty}{\underset{j=1}{\bigoplus}}~\mathcal{U}_j.\]

\noindent Consequently it is straight forward that $\lambda$ be an eigenvalue of the operator $\mathcal{L}$, if and only if  $\lambda$ be an eigenvalue of each of the matrix $$\mathcal{L}_j=J(E)-\mu_{j}\mathcal{D}=\left[
    \begin{array}{ccc}
    r\left(1-\frac{2u_{eq}}{X_0}\right)-f'(u_{eq})v_{eq}-\mu_{j} & -f(u_{eq}) \\
    f'(u_{eq})v_{eq} & f(u_{eq})-1-2k_3v_{eq}-d\mu_{j}
    \end{array}\right],$$ for some $j\geq 0$. Hence, it is sufficient to discuss about the eigenvalues of each of the matrix $\mathcal{L}_j$. The characteristic equation corresponding to the matrix $\mathcal{L}_j$ be
\[det(\mathcal{L}_j-\lambda I_n)=\lambda^2-\Phi_{j}\lambda+\Psi_{j}=0,~~j=0,1,2,\ldots,\infty,\]
where \[\Phi_{j}=r\left(1-\frac{2u_{eq}}{X_0}\right)-(f'(u_{eq})+2k_3)v_{eq} + f(u_{eq})-1-(1+d)\mu_j,~\mbox{and}\]
\[\Psi_{j}=d\mu_j^2-\mu_j\left(d\left[r\left(1-\frac{2u_{eq}}{X_0}\right)-f'(u_{eq})v_{eq}\right]+f(u_{eq})-1-2k_3v_{eq}\right)+\det(J(E)).\]
Now evaluating the trace and determinant for all values of $j$ at different spatially homogeneous steady states, we obtain the following results.
\begin{thm}
    (a) The constant steady state $E_0=(0,0)$ of the system (\ref{alg_eqn_pde9}) is unstable. (b) The constant steady state $E_1=(X_0,0)$ of the system (\ref{alg_eqn_pde9}) is asymptotically stable whenever the condition \textit{\textbf{C1:}}$f(X_0)<1$ holds, and unstable otherwise.
    
\end{thm}

\begin{proof}
    (a) Evaluating at the steady state $E_0=(0,0)$, we obtain $\Phi_{j}=r-1-(1+d)\mu_j~\mbox{and}~\Psi_{j}=d\mu_j^2-\mu_j(rd-1)-r,$ for $j=0,1,2,\cdots,\infty$. Observe that $\Psi_{0}=-r<0$, which implies this constant steady state is unstable.  

    \noindent (b) Evaluating at the steady state $E_1=(X_0,0)$, we obtain $\Phi_{j}=-r+f(X_0)-1-(1+d)\mu_j$ and $\Psi_{j}=d\mu_j^2-\mu_j(-rd+f(X_0)-1)-r(f(X_0)-1)$. Whenever we have $f(X_0)<1$, then $\Phi_{j}<0$ and $\Psi_{j}>0$ for all $j=0,1,2,\cdots,\infty.$ Hence, we can conclude that the constant homogeneous steady state $E_1=(X_0,0)$ is asymptotically stable if $f(X_0)<1$. Otherwise note that $\Psi_{0}=-r(f(X_0)-1)<0$, which implies that the steady state is unstable.
\end{proof}

Furthermore, if the parameter values of the functional response are chosen such that the inequality $f(X_0)<1$ holds, then the steady state $E_1=(X_0,0)$ becomes globally asymptotically stable. Consequently, pattern formation would not be observed under this parameter restriction. This is analytically proven as follows:
\begin{thm}
    If functional response satisfies the condition \textit{\textbf{C1}}, then any global solution $\left(u(\mathbf{x},t),v,(\mathbf{x},t)\right)$ for the system (\ref{alg_eqn_pde9}) satisfies \[\underset{t \rightarrow \infty}{\lim}\mbox{ } u(\mathbf{x},t)=X_0\mbox{,~}\underset{t \rightarrow \infty}{\lim}\mbox{ } v(\mathbf{x},t)=0 \mbox{~uniformly on }\Bar{\Omega}.\]
    That is, $E_1=(X_0,0)$ is globally asymptotically stable in $\mathbb{R}^2_{\geq 0}.$
\end{thm}
\begin{proof}
    Following the proof of Theorem \ref{thm2}(a), we can say that, whenever $f(X_0)<1$, $\hat{v}(t)\to 0$ as $t\to \infty$, where we have the inequality $v(\mathbf{x},t)<\hat{v}(t)$ for all $\mathbf{x}\in \Omega$. Thus, $\underset{t \rightarrow \infty}{\lim}~ v(\mathbf{x},t)=0 \mbox{~uniformly on }\Bar{\Omega}$.

    Now, for significantly small $\epsilon>0$, we can find suitable $T_3\in (0,\infty)$ such that $v(\mathbf{x},t)<\epsilon,\forall \mathbf{x}\in \Omega, t\geq T_3.$ Therefore we obtain the inequality 
    \[\,\,\,\,\,\,\,\,\,\,\,\,\,\,\,\,\,\,\,\,\,\,\,\,\,\,\,\,\,\,\,\,\,\,\,\,\frac{\partial u(\mathbf{x},t)}{\partial t}-\Delta u(\mathbf{x},t)-ru(\mathbf{x},t)\left(1-\frac{u(\mathbf{x},t)}{X_0}\right)+f^{\infty}u(\mathbf{x},t)\epsilon\]
    \[\geq0=\frac{\partial \hat{z}(t)}{\partial t}-\Delta \hat{z}(t)-r\hat{z}(t)\left(1-\frac{\hat{z}(t)}{X_0}\right)+f^{\infty}u(\mathbf{x},t)\epsilon,\,\,\,\,\,\,\,\]
    together with the inequality of initial condition
\[u(\mathbf{x},T_3)\geq \min_{\mathbf{x} \in \Bar{\Omega}}u(\mathbf{x},T_3)=\hat{z}(T_3),\]
and boundary condition 
\[\frac{\partial u(\mathbf{x},t)}{\partial n}=0\geq 0=\frac{\partial \hat{z}(t)}{\partial n},\]
where $w(t)$ is the solution to the initial value problem
\begin{equation}
   \begin{cases}
     \frac{d\hat{z}}{dt}=\hat{z}(t)\left[r\left(1-\frac{\hat{z}(t)}{X_0}\right)-f^{\infty}\epsilon\right],~ t>T_3,\\
     \hat{z}(T_3)=\underset{\mathbf{x} \in supp(u(\mathbf{x},T_3))}{\min}u(\mathbf{x},T_3)>\underset{\mathbf{x} \in supp(u(\mathbf{x},0))}{\min}u(\mathbf{x},0)\geq 0.
   \end{cases}\,
\end{equation}
Therefore, using the comparison principle, we can have the inequality $v(\mathbf{x},t)\geq \hat{z}(t)$ \cite{pao}. Now observe that $\underset{t \rightarrow \infty}{\lim} \hat{z}(t)=\frac{X_0}{r}(r-f^{\infty}\epsilon)$, using arbitrariness of $\epsilon$, we can conclude that \[\underset{t \rightarrow \infty}{\liminf} \min_{\mathbf{x} \in \Bar{\Omega}}u(\mathbf{x},t)\geq X_0.\]
This along with the result of Theorem \ref{thm2}(b) implies that $\underset{t \rightarrow \infty}{\lim}\mbox{ } u(\mathbf{x},t)=X_0\mbox{~uniformly on }\Bar{\Omega}$.
\end{proof}

From our discussion of the temporal system, we already know that system \ref{alg_eqn_pde9} can exhibit at most three nontrivial homogeneous coexisting steady states, depending on the choice of parameter values. Next, we examine the stability of these steady states under spatially heterogeneous perturbations. Rather than analyzing all possible scenarios individually, we focus on the conditions under which a steady state remains stable in the presence of small-amplitude heterogeneous perturbations. For convenience, let us consider..
\[\mathcal{A}(d,\mu)=d\mu^2-\left(d\left[r\left(1-\frac{2u}{X_0}\right)-\frac{f'(u_*)\left(f(u_*)-1\right)}{k_3}\right]+1-f(u_*)\right)\mu+\left(1-f(u_*)\right)G(u_*),\]
\noindent and denote the discriminant of the quadratic equation $\mathcal{A}(d,\mu)=0$ by:
\[\mathcal{D}(d)=\left(d\left[r\left(1-\frac{2u}{X_0}\right)-\frac{f'(u_*)\left(f(u_*)-1\right)}{k_3}\right]+1-f(u_*)\right)^2+4d\left(f(u_*)-1\right)G(u_*).\] 
Note that if the discriminant is positive, the equation $\mathcal{A}(d,\mu)=0$ has two real roots, which are of the following form:
\[\mu^{\pm}(d)=\frac{\left(d\left[r\left(1-\frac{2u}{X_0}\right)-\frac{f'(u_*)\left(f(u_*)-1\right)}{k_3}\right]+1-f(u_*)\right)\pm\sqrt{\mathcal{D}(d)}}{2d}.\]
\noindent Assume two real roots of $\mathcal{A}(d,\mu)=0$ are feasible, in that case, let us define the set \[\mathcal{S}(d):=\left(\mu^-,\mu^+\right)-\mathbb{R}_{< 0}.\]
\noindent Since our focus is on the stability of coexisting homogeneous steady states under small-amplitude spatially heterogeneous perturbations, we consider only those steady states that are already stable under spatially homogeneous perturbations. For clarity in the upcoming discussion, let us denote the following expressions of $u$ as
\begin{align*}
    &D_1(u):= (f(u)-1)\left\{1+\frac{f'(u)}{k_3}\right\},\,\,\,\,\,D_2(u):= (f(u)-1)\left\{\frac{1}{d}+\frac{f'(u)}{k_3}\right\},\,\,\,\\ &D_3(u):= 2\sqrt{d(1-f(u))G(u)}.
\end{align*}
We now turn our attention to the local stability of coexisting homogeneous steady states under spatially heterogeneous perturbations.

\begin{thm}  The coexisting homogeneous steady-state $E_*=(u_*,v_*)$ is locally asymptotically stable

\noindent(i) whenever the prey density, $u_*$, satisfies the condition: 
 \begin{equation}\label{cond}
r\left(1-\frac{2u_*}{X_0}\right)<f'(u_*)\frac{\left(f(u_*)-1\right)}{k_3};
\end{equation}
\noindent(ii) whenever the prey density, $u_*$, does not satisfy the condition \ref{cond} but satisfy the inequality 
       \[r\left(1-\frac{2u_*}{X_0}\right)<min\{D_1(u_*),D_2(u_*)\};\]
\noindent(iii) whenever the prey density, $u_*$, does not satisfy the condition \ref{cond} but satisfy the inequality 
       \[D_2(u_*)<r\left(1-\frac{2u_*}{X_0}\right)<min\{D_1(u_*),D_2(u_*)+D_3(u_*)\};\]
\noindent(iv) whenever the prey density, $u_*$, does not satisfy the condition \ref{cond} but satisfy the inequality 
       \[D_2(u_*)+D_3(u_*)<r\left(1-\frac{2u_*}{X_0}\right)<D_1(u_*)\]
       and additionally satisfy the condition $\overline{\mathcal{S}(d)}\cap \mathcal{S}_{\Delta}=\phi$, where $\overline{\mathcal{S}(d)}$ represents closure of the set $\mathcal{S}(d)$.

\end{thm}

\begin{proof} The characteristic equations corresponding to the matrix $\mathcal{L}_j$ evaluated at $E_*=(u_*,v_*)$ is given by
\[\lambda^2-\Phi_{_j}\lambda+\Psi_{j}=0,~~j=0,1,2,\cdots,\infty,\]
where \[\Phi_{j}=r\left(1-\frac{2u_*}{X_0}\right)-D_1(u_*)-(1+d)\mu_j,~\mbox{and}\]
\[\Psi_{j}=d\mu_j^2-d\left[r\left(1-\frac{2u_*}{X_0}\right)-D_2(u_*)\right]\mu_j+\det(J(E_*)).\]

\noindent(i) If Condition \ref{cond} holds, it is straightforward to check that $\Phi_{j}<0$ and $\Psi_{j}>0$  $\forall~j=0,1,2,\cdots, \infty$. Hence, both eigenvalues of each matrix $\mathcal{L}_j$ have negative real parts, ensuring that all elements of the spectrum of $\mathcal{L}$ be negative. Thus, we can conclude that the coexisting homogeneous steady state that meets the condition \ref{cond} is locally asymptotically stable.

\noindent(ii) If the condition \ref{cond} does not hold but the condition $r\left(1-\frac{2u_*}{X_0}\right)<min\{D_1(u_*),D_2(u_*)\}$ is satisfied, then $\Phi_{j}<0$ and $\Psi_{j}>0$ for all $j=0,1,2,,\cdots, \infty$. Hence both eigenvalues of each matrix $\mathcal{L}_j$ have negative real parts, and by the similar argument of the case (i), the coexisting homogeneous steady-state is locally asymptotically stable.

\noindent(iii) If the condition \ref{cond} does not hold but the condition $D_2(u_*)<r\left(1-\frac{2u_*}{X_0}\right)<min\{D_1(u_*),D_2(u_*)+D_3(u_*)\}$ is satisfied, then a straightforward observation shows that $\Phi_{j}<0$ and $\Psi_{j}>0$, for all $j=0,1,2,\cdots, \infty$. Therefore, for each $j\geq 0$, both the eigenvalues of the matrices $\mathcal{L}_j$ have negative real parts. Using a similar argument as in case (i), we conclude that the coexisting homogeneous steady state is locally asymptotically stable in this scenario.

\noindent(iv) If the condition \ref{cond} does not hold but the condition $D_2(u_*)+D_3(u_*)<r\left(1-\frac{2u_*}{X_0}\right)<D_1(u_*)$ is satisfied, then note that $\mathcal{D}(d)>0$. Consequently, we can compute $\mu^{\pm}(d)$, ensuring that the set $\mathcal{S}(d)$ is well-defined in this case. Furthermore, if the condition $\overline{\mathcal{S}(d)}\cap \mathcal{S}_{\Delta}=\phi$ holds for some choice of parameter values, then for all $j \geq 0$, both the eigenvalues of the matrices $\mathcal{L}_j$ will be negative under this parameter restriction. Hence, by a similar argument as in case (i), we conclude that the coexisting homogeneous steady-state is locally asymptotically stable.
\end{proof}

\noindent From the above theorem, it is easy to observe that if $d\leq 1$, then $D_1(u_*)\leq D_2(u_*)$. This implies that the coexisting homogeneous steady state of system \ref{alg_eqn_pde9} exhibits the same local stability behavior under spatially homogeneous and heterogeneous perturbations. Moreover, we can further demonstrate that a homogeneous coexisting steady state satisfying Condition \ref{cond} is globally asymptotically stable. This result implies that system \ref{alg_eqn_pde9} does not admit any heterogeneous steady state under this parameter restriction, regardless of the values of the diffusion coefficients.

\begin{thm}
    If a coexisting homogeneous steady state satisfies condition \ref{cond}, it is globally asymptotically stable.
\end{thm}

\begin{proof} Consider $\left(u(\mathbf{x},t),v(\mathbf{x},t)\right)$ as the solution of the system (\ref{alg_eqn_pde9}) and $E_*=(u_*,v_*)$ be the coexisting homogeneous steady state satisfying the condition \ref{cond}. Let us consider the Lyapunov function 
\[E\left(u(\mathbf{x},t),v(\mathbf{x},t)\right)=\underset{\Omega~~}{\int}\overset{~~u}{\underset{u_*~}{\int}}\frac{f(\zeta)-f(u_*)}{f(\zeta)}d\zeta dx+\underset{\Omega~~}{\int}\overset{~~v}{\underset{v_*~}{\int}}\frac{\nu-v_*}{\nu}d\nu dx.\]
Then \[\frac{d E}{d t}=\underset{\Omega~~}{\int}\left[\frac{f(u)-f(u_*)}{f(u)}\frac{\partial u}{\partial t}+\frac{v-v_*}{v}\frac{\partial v}{\partial t}\right]dx,\]
\[=I_1+I_2,~~~~~~~~~~~~~~~~~~~~~~~~~~~~~~~~~~~\]
where $$I_1=\underset{\Omega~~}{\int}\left[\frac{f(u)-f(u_*)}{f(u)}\left\{ru\left(1-\frac{u}{X_0}\right)-f(u)\frac{f(u_*)-1}{k_3}\right\}-k_3(v-v_*)^2\right]~\mbox{and}$$ $$I_2=-\left[\underset{\Omega~~}{\int}\frac{f'(u)f(u_*)}{f^2(u)}|\nabla u|^2dx+d\underset{\Omega~~}{\int}\frac{v_*}{v^2}|\nabla v|^2dx\right].$$
Clearly, $I_2<0$ for any solution $\left(u(\mathbf{x},t),v(\mathbf{x},t)\right)$. If the coexisting steady state satisfy the condition \ref{cond}, then the expression $ru\left(1-\frac{u}{X_0}\right)-f(u)\frac{f(u_*)-1}{k_3}$ become a decreasing function of $u$. Irrespective of the magnitude of $u$ is greater or less than $u_*$, the expression $\frac{f(u)-f(u_*)}{f(u)}\left\{ru\left(1-\frac{u}{X_0}\right)-f(u)\frac{f(u_*)-1}{k_3}\right\}$ will always negative. Therefore, we can say that $\frac{d E}{dt}<0$ along any orbit $\left(u(\mathbf{x},t),v(\mathbf{x},t)\right)$ of the system (\ref{alg_eqn_pde9}) with positive initial condition $(u_0,v_0)$ and $\frac{d E}{dt}=0$ only if $\left(u(\mathbf{x},t),v(\mathbf{x},t)\right)=(u_*,v_*)$. Hence, it is established that the coexisting homogeneous steady-state is globally asymptotically stable whenever the prey density $u_*$ satisfies the condition \ref{cond}. 
\end{proof}

\subsection{Heterogeneous positive steady state}\label{HetSteb} 
In this section, we study the criteria for the existence and non-existence of diffusion-driven spatially heterogeneous positive steady states for system (\ref{alg_eqn_pde9}). Positive steady states are the positive solutions of the following equations:
 \begin{subequations}\label{alg_eqn_pde2}
\begin{eqnarray}
-\Delta u&=&ru\left(1-\frac{u}{X_0}\right)-f\left(u\right)v,\\
-d\Delta v&=&f\left(u\right)v-v-k_3{v}^2,
\end{eqnarray}
\end{subequations}
where $u=u(\mathbf{x})~\mbox{and}~v=v(\mathbf{x})$, along with the no-flux boundary conditions $\frac{\partial u(\mathbf{x})}{\partial n}=\frac{\partial v(\mathbf{x})}{\partial n}=0$ for $\mathbf{x}\in \partial\Omega$, where $n$ is the outward unit normal to $\partial \Omega$.  We begin our discussion by establishing upper and lower bounds for the positive solutions of system (\ref{alg_eqn_pde2}). Following this, we prove the existence and non-existence of non-constant positive steady-state solutions.
\subsubsection{Bounds for positive steady states}
To establish upper and lower bounds for the positive solutions of system (\ref{alg_eqn_pde2}), we use the Harnack inequality \cite{LIN19881} up to the boundary and a proposition based on the Maximum Principle \cite{LOU199679}. For clarity and smooth readability, we restate both statements below:

\noindent \textbf{Proposition} (\textit{Maximum Principle \cite{LOU199679}}): Suppose $g\in C(\overline{\Omega}\times\mathbb{R})$ where $\Omega \in \mathbb{R}^n$, then 
\noindent(i) if $w(\mathbf{x})\in C^2(\Omega)\cap C^1(\overline{\Omega})$ and satisfies the condition \[\Delta w(\mathbf{x})+g(x,w(\mathbf{x}))\geq 0 ~\mbox{in}~\Omega,~ \frac{\partial w}{\partial \eta}\leq 0~\mbox{on}~\partial\Omega,\]
    and $w(x_0)=\underset{\overline{\Omega}}{\max}~w(\mathbf{x})$, then $g(x_0,w(x_0))\geq 0$.
 \noindent(ii) if $w(\mathbf{x})\in C^2(\Omega)\cap C^1(\overline{\Omega})$ and satisfies the condition \[\Delta w(\mathbf{x})+g(x,w(\mathbf{x}))\leq 0 ~\mbox{in}~\Omega,~ \frac{\partial w}{\partial \eta}\geq 0~\mbox{on}~\partial\Omega,\]
    and $w(x_0)=\underset{\overline{\Omega}}{\min}~w(\mathbf{x})$, then $g(x_0,w(x_0))\leq 0$.

\noindent Harnack inequality \cite{LIN19881} \textbf{:} Suppose $\Omega \in \mathbb{R}^n$ and $w(\mathbf{x})\in C^2(\Omega)\cap C^1(\overline{\Omega})$ be a positive solution to $d\Delta w(\mathbf{x})+c(\mathbf{x})w=0$ in $\Omega$ subjected to the no-flux boundary condition, where $d$ is positive constant and $c(\mathbf{x})\in C(\overline{\Omega})$, then there exist a positive constant $C_{\#}=C_{\#}(n,\Omega,d,\|c\|_{\infty})$ such that for any ball $B(z,R)$ of radius $R$ and center at $z\in \overline{\Omega}$, we have 
\[\underset{\mathbf{x}\in B(z,R)\cap \Omega}{\sup}~w(\mathbf{x})~\leq~ C_{\#}\underset{\mathbf{x}\in B(z,R)\cap \Omega}{\inf}~w(\mathbf{x}).\]

\noindent Now, we define and prove the main theorem, which provides an estimate for the upper and lower bounds of the solution to system (\ref{alg_eqn_pde2}).

\begin{thm}\label{Boundss}
If the choice of parameter values implies the inequality $f(X_0)>1$, then any positive solution $(u(\mathbf{x}),v(\mathbf{x}))$ of the system (\ref{alg_eqn_pde2}) is bounded above by  $u(\mathbf{x})\leq X_0$ and $v(\mathbf{x}) \leq \frac{f(X_0)-1}{k_3}$. Moreover, for some fixed $\tilde{d}$, there exist a constant $C$ depending upon $\Omega~\mbox{and}~ \tilde{d}$ such that $u(\mathbf{x})\geq C$ and $v(\mathbf{x}) \geq C$ for $d\geq \tilde{d}$.
\end{thm}

\begin{proof}
    Let us assume that $(u(\mathbf{x}),v(\mathbf{x}))$ is a solution to the system (\ref{alg_eqn_pde2}). Suppose $x_0,x_1\in \overline{\Omega}$ be two points such that $u(x_0)=\underset{\mathbf{x}\in \overline{\Omega}}{\max}~ u(\mathbf{x})$ and $v(x_1)=\underset{\mathbf{x}\in \overline{\Omega}}{\max} ~v(\mathbf{x})$. Using the above proposition, from the first equation of the system (\ref{alg_eqn_pde2}), we have 
    \[r\left(1-\frac{u(x_0)}{X_0}\right)-\frac{f\left(u(x_0)\right)}{u(x_0)}v(x_0)\geq 0.\]
    Therefore we can conclude that $u(x_0)\leq X_0$, which further implies that $u(\mathbf{x})\leq X_0$ for all $\mathbf{x}\in \Omega$. Similarly, from the second equation of the system (\ref{alg_eqn_pde2}), we get 
    \[f(u(x_1))-1-k_3v(x_1)\geq 0.\]
    So, we can easily conclude that $v(\mathbf{x})\leq \frac{f(u(x_1))-1}{k_3}\leq  \frac{f(X_0)-1}{k_3}$. Thus, we obtain the upper bound for the solution of the system (\ref{alg_eqn_pde2}). Next, we will prove that the solution must have a lower bound. Note that 
    \[\|r\left(1-\frac{u(\mathbf{x})}{X_0}\right)-\frac{f\left(u(\mathbf{x})\right)}{u(\mathbf{x})}v(\mathbf{x})\|_{\infty}\leq r~\mbox{and}~\|f(u(\mathbf{x}))-1-k_3v(\mathbf{x})\|_{\infty}\leq \left(f(X_0)-1\right).\]
    Hence, if we choose some suitable $d > \tilde{d}$, using Harnack inequality we can find some $C_{*}$ depending upon $n, \tilde{d},\max\{r,\left( f(X_0)-1\right)\} ~\mbox{and}~\Omega$, which satisfies the following inequalities
    \[\underset{\mathbf{x}\in \overline{\Omega}}{\sup}~u(\mathbf{x})~\leq~ C_{*}\underset{\mathbf{x}\in \overline{\Omega}}{\inf}~u(\mathbf{x}),~\mbox{and}~\underset{\mathbf{x}\in \overline{\Omega}}{\sup}~v(\mathbf{x})~\leq~ C_{*}\underset{\mathbf{x}\in \overline{\Omega}}{\inf}~v(\mathbf{x}),\]
    provided $d > \tilde{d}$. Suppose the solution has no lower bound, then there exist a sequence $\{d_j\}_{j=1}^{\infty}$ satisfying $d_j> \tilde{d}$, and corresponding positive solutions $\{\left(u_j(\mathbf{x}),v_j(\mathbf{x})\right)\}_{j=1}^{\infty}$ of system (\ref{alg_eqn_pde2}), such that 
    \[\underset{j\to \infty}{\lim} ~\underset{\mathbf{x}\in \overline{\Omega}}{\max}~u_j(\mathbf{x})=0,~\mbox{or}~\underset{j\to \infty}{\lim} ~\underset{\mathbf{x}\in \overline{\Omega}}{\max}~v_j(\mathbf{x})=0.\]
    Since each $\left(u_j(\mathbf{x}),v_j(\mathbf{x})\right)$ satisfy the system of equation (\ref{alg_eqn_pde2}), by integrating both side we obtain 
    \[\underset{\Omega}{\int}\left[ru_j(\mathbf{x})\left(1-\frac{u_j(\mathbf{x})}{X_0}\right)-f(u_j(\mathbf{x}))v_j(\mathbf{x})\right]d\mathbf{x}=0,\]\[\underset{\Omega}{\int}\left[f(u_j(\mathbf{x}))v_j(\mathbf{x})-v_j(\mathbf{x})-k_3v^2_j(\mathbf{x})\right]d\mathbf{x}=0.~~~~~~\]
   Using the regularity theory of elliptic equations \cite{EPDE}, there exists a sub-sequence of $\{\left(u_j(\mathbf{x}),v_j(\mathbf{x})\right)\}_{j=1}^\infty$, which we still denote by $\{\left(u_j(\mathbf{x}),v_j(\mathbf{x})\right)\}_{j=1}^\infty$, and two non-negative functions $(\tilde{u}(\mathbf{x}),\tilde{v}(\mathbf{x}))\in [C^2(\Omega)]^2$ such that $\left(u_j,v_j\right)\to (\tilde{u},\tilde{v})$ in $[C^2(\Omega)]^2$ as $j \to \infty$. Now, our assumption of not having a lower bound concludes that either $\tilde{u}\equiv 0$ or $\tilde{v}\equiv 0$. Together with this, we obtain the following integration
   \[\underset{\Omega}{\int}\left[r\tilde{u}(\mathbf{x})\left(1-\frac{\tilde{u}(\mathbf{x})}{X_0}\right)-f(\tilde{u}(\mathbf{x}))\tilde{v}(\mathbf{x})\right]d\mathbf{x}=0,\]\[\underset{\Omega}{\int}\left[f(\tilde{u}(\mathbf{x}))\tilde{v}(\mathbf{x})-\tilde{v}(\mathbf{x})-k_3\tilde{v}^2(\mathbf{x})\right]d\mathbf{x}=0,~~~~~~\]
   as $j \to \infty$. Now, we draw the conclusion by considering the following three cases.
   
  \noindent(i) If $\tilde{u}\equiv 0$ and $\tilde{v}\neq 0$, then from the inequality: $\underset{\mathbf{x}\in \overline{\Omega}}{\sup}~\tilde{v}(\mathbf{x})~\leq~ C_{*}\underset{\mathbf{x}\in \overline{\Omega}}{\inf}~\tilde{v}(\mathbf{x})$, we can conclude that $\tilde{v}> 0$ on $\Omega$. Therefore, $\underset{\Omega}{\int}\left[f(\tilde{u}(\mathbf{x}))\tilde{v}(\mathbf{x})-\tilde{v}(\mathbf{x})-k_3\tilde{v}^2(\mathbf{x})\right]d\mathbf{x}=\underset{\Omega}{\int}\left[-\tilde{v}(\mathbf{x})-k_3\tilde{v}^2(\mathbf{x})\right]d\mathbf{x}<0$ since $\tilde{v}(\mathbf{x})$ is positive, which is a contradiction.
 
 \noindent(ii) From the inequality: $\underset{\mathbf{x}\in \overline{\Omega}}{\sup}~\tilde{u}(\mathbf{x})~\leq~ C_{*}\underset{\mathbf{x}\in \overline{\Omega}}{\inf}~\tilde{u}(\mathbf{x})$, we can conclude that $0<\tilde{u}\leq X_0$ on $\overline{\Omega}$ if $\tilde{u}\neq 0$ and $\tilde{v}\equiv 0$, which further implies $\tilde{u}\equiv X_0$. Since we have \[\underset{\Omega}{\int}\left[f(u_j(\mathbf{x}))v_j(\mathbf{x})-v_j(\mathbf{x})-k_3v^2_j(\mathbf{x})\right]d\mathbf{x}=0,\] we can find $\mathbf{x}_j\in \Omega$, assuming the sequence $\{\mathbf{x}_j\}_{j=1}^\infty$ converges to a limit within $\overline{\Omega}$, such that
\[f(u_j(\mathbf{x}_j))v_j(\mathbf{x}_j)-v_j(\mathbf{x}_j)-k_3v^2_j(\mathbf{x}_j)=0.\]
This further implies that $f(u_j(\mathbf{x}_j))-1-k_3v_j(\mathbf{x}_j)=0$. Taking the limit $j \to \infty$, we obtain $f(X_0)=1$, which is a contradiction to the assumption $f(X_0)>1$. 

\noindent(iii) If $\tilde{u}\equiv 0$ and $\tilde{v}\equiv 0$, then using similar argument as above, we can find a sequence $\{\mathbf{x}_j\}_{j=1}^{\infty}\in \Omega$ converges to a limit within $\overline{\Omega}$, such that
\[f(u_j(\mathbf{x}_j))v_j(\mathbf{x}_j)-v_j(\mathbf{x}_j)-k_3v^2_j(\mathbf{x}_j)=0\]
holds. This further implies that $f(u_j(\mathbf{x}_j))-1-k_3v_j(\mathbf{x}_j)=0$. Taking the limit $j \to \infty$, we obtain $-1=0$, which is a contradiction.

   Therefore, the solution of the system (\ref{alg_eqn_pde2}) must have some lower bound. 
\end{proof}

\subsubsection{Non-existence of positive heterogeneous steady state}
 Now, we determine the conditions that establish the non-existence of positive heterogeneous solutions of system (\ref{alg_eqn_pde2}), which implies the non-existence of positive heterogeneous steady states for system (\ref{alg_eqn_pde9}).
\begin{thm}\label{non-existence}
    If the condition $f(X_0)>1$ holds and the parameter value of $r$ is chosen so that $r < \mu_1$, where $\mu_1$ is the smallest non-zero eigenvalue of the operator $-\Delta$, then there exists a positive $d_*$, depending upon the set $\Omega$ and parameter values involved with the system, such that the system (\ref{alg_eqn_pde2}) does not have any positive solution.
\end{thm}

\begin{proof}
    Consider $(u(\mathbf{x}),v(\mathbf{x}))$ to be a positive non-constant solution of the system (\ref{alg_eqn_pde2}). Using the notation $\overline{u}=\frac{1}{|\Omega|}\int_{\Omega}u(\mathbf{x})d\mathbf{x}$ and $\overline{v}=\frac{1}{|\Omega|}\int_{\Omega}v(\mathbf{x})d\mathbf{x}$ to represent the spatial averages of $u(\mathbf{x})~\mbox{and}~v(\mathbf{x})$, we can easily derive the following result from this representation:
    \[\int_{\Omega}(u-\overline{u})d\mathbf{x}=0~\mbox{and}~\int_{\Omega}(v-\overline{v})d\mathbf{x}=0.\]
    If we Multiply both side of the first equation of the system (\ref{alg_eqn_pde2}) by $(u-\overline{u})$ and integrate over $\Omega$, using the no-flux boundary condition, we get
    \begin{align*}
        \underset{\Omega}{\int}|\nabla(u-\overline{u})|^2d\mathbf{x} &=\underset{\Omega}{\int}(u-\overline{u})F_1(u,v)d\mathbf{x}=\underset{\Omega}{\int}(u-\overline{u})\left(F_1(u,v)-F_1(\overline{u},\overline{v})\right)d\mathbf{x}\\
        &=\underset{\Omega}{\int}(u-\overline{u})^2\left\{r-\frac{r}{X_0}(u+\overline{u})\right\}d\mathbf{x}-\underset{\Omega}{\int}(u-\overline{u})\left\{f(u)v-f(\overline{u})\overline{v}\right\}d\mathbf{x}\\
        &\leq r\underset{\Omega}{\int}(u-\overline{u})^2d\mathbf{x}-f(\overline{u})\underset{\Omega}{\int}(u-\overline{u})(v-\overline{v})d\mathbf{x}.
    \end{align*}

    \noindent Similarly, multiplying the second equation of the system (\ref{alg_eqn_pde2}) by $(v-\overline{v})$ and integrating over $\Omega$ using the no-flux boundary condition, we obtain
    \begin{align*}
        d\underset{\Omega}{\int}|\nabla(v-\overline{v})|^2d\mathbf{x}&=\underset{\Omega}{\int}(v-\overline{v})F_2(u,v)d\mathbf{x}=\underset{\Omega}{\int}(v-\overline{v})\left(F_2(u,v)-F_2(\overline{u},\overline{v})\right)d\mathbf{x}\\
        &\leq \underset{\Omega}{\int}(v-\overline{v})\left[f(u)v-f(\overline{u})\overline{v}\right]d\mathbf{x}- \underset{\Omega}{\int}(v-\overline{v})^2\left\{1+k_3(v+\overline{v})\right\}d\mathbf{x}\\
        &\leq (f^{\infty}-1)\underset{\Omega}{\int}(v-\overline{v})^2d\mathbf{x}+\beta_1\overline{v}\underset{\Omega}{\int}(u-\overline{u})(v-\overline{v})d\mathbf{x},
    \end{align*}
where $\beta_1=\lim_{u\to 0}\frac{f(u)}{u}<\infty$. Adding both side of above two inequalities, we obtain
\[\underset{\Omega}{\int}|\nabla(u-\overline{u})|^2d\mathbf{x}+d\underset{\Omega}{\int}|\nabla(v-\overline{v})|^2d\mathbf{x}~~~~~~~~~~~~~~~~~~~~~~~~~~~~~~~~\]\[\leq r\underset{\Omega}{\int}(u-\overline{u})^2d\mathbf{x}+(f(X_0)-1)\underset{\Omega}{\int}(v-\overline{v})^2d\mathbf{x}+L\underset{\Omega}{\int}(u-\overline{u})(v-\overline{v})d\mathbf{x},\]
where $L=(\beta_1\overline{v}-f(\overline{u}))$. Now if we apply $\epsilon$-Young inequality, the above inequality becomes
\[\underset{\Omega}{\int}|\nabla(u-\overline{u})|^2d\mathbf{x}+d\underset{\Omega}{\int}|\nabla(v-\overline{v})|^2d\mathbf{x}~~~~~~~~~~~~~~~~~~~~~~~~~~~~\]\[\leq \left(r+\frac{L}{2\epsilon}\right)\underset{\Omega}{\int}(u-\overline{u})^2d\mathbf{x}+\left\{(f(X_0)-1)+\frac{L\epsilon}{2}\right\}\underset{\Omega}{\int}(v-\overline{v})^2d\mathbf{x},\]
where $\epsilon$ is an arbitrary positive constant. Since $\mu_1$ is the minimal eigenvalue of $-\Delta$, using Poincare inequality \cite{Leoni2009AFC}, finally we obtain
\[\underset{\Omega}{\int}|\nabla(u-\overline{u})|^2d\mathbf{x}+d\underset{\Omega}{\int}|\nabla(v-\overline{v})|^2d\mathbf{x}~~~~~~~~~~~~~~~~~~~~~~~~~~~~~~~~~~~~~~~~~~~~~\]\[\leq \frac{1}{\mu_1}\left(\left(r+\frac{L}{2\epsilon}\right)\underset{\Omega}{\int}|\nabla(u-\overline{u})|^2d\mathbf{x}+\left\{(f(X_0)-1)+\frac{L\epsilon}{2}\right\}\underset{\Omega}{\int}|\nabla(v-\overline{v})|^2d\mathbf{x}\right).\]
Using the assumption $r<\mu_1$, we can establish the inequality $\frac{1}{\mu_1}\left(r+\frac{L}{2\epsilon}\right)<1$, by appropriately choosing a suitable value of $\epsilon$.Next, by selecting $d_*\leq d$ with $d_*=\frac{1}{\mu_1}\left\{(f(X_0)-1)+\frac{L\epsilon}{2}\right\}$, we obtain $\nabla(u-\overline{u})=\nabla(v-\overline{v})=0$, which implies that the positive solution of the system (\ref{alg_eqn_pde2}) is homogeneous. This completes the proof.
\end{proof}

\subsubsection{Existence of heterogeneous positive steady state} 
Now we derive the conditions for existence of positive heterogeneous solution for the system (\ref{alg_eqn_pde2}).
Let us define $W(\mathbf{x})=(u(\mathbf{x}),v(\mathbf{x}))^T$, $W_*=(u_*,v_*)^T$, $\mathcal{U}^+=\left\{(u,v)\in \mathcal{U}| u>0, v>0~\mbox{for}~\mathbf{x}\in \overline{\Omega}\right\}$ and $\mathcal{F}(W)=(F_1,d^{-1}F_2)^T$, and rewrite the system (\ref{alg_eqn_pde2}) as follows: 
 \begin{subequations}\label{alg_eqn_pde3}
\begin{eqnarray}
-\Delta W(\mathbf{x})&=&\mathcal{F}(W),\\
\frac{\partial W(\mathbf{x})}{\partial n}&=&0.
\end{eqnarray}
\end{subequations}
$\textbf{W}$ will be a positive solution of the system (\ref{alg_eqn_pde3}) if and only if it satisfies 
\[\mathcal{T}(\textbf{W}):=\textbf{W}-(I-\Delta)^{-1}\left[\textbf{W}+\mathcal{F}(\textbf{W})\right]=0~\mbox{for}~\textbf{W}\in \mathcal{U}^+,\]
where $I$ is the identity operator and $(I-\Delta)^{-1}$ is the inverse operator of $(I-\Delta)$, subject to the zero-flux boundary condition. Thus, $\mathcal{T}(.)$ becomes a compact perturbation of the identity operator. Note that 
\[\mathcal{T}_W(W^*):=I-(I-\Delta)^{-1}\left[I+\mathcal{F}_W(W_*)\right]\]
with $$\mathcal{F}_W(W_*)=\left[
    \begin{array}{ccc}
    r(1-\frac{2u_*}{X_0})-f'(u_*)\frac{f(u_*)-1}{k_3} & -f(u_*)\\
    f'(u_*)\frac{f(u_*)-1}{dk_3} & \frac{1-f(u_*)}{d}
    \end{array}\right].$$Now, our next target is to study the eigenvalue of $\mathcal{T}_W(W^*)$. If $\lambda$ be an eigenvalue of $\mathcal{T}_W(W^*)$ and $\psi=\left(\psi_1,\psi_2\right) \in \mathcal{U}^+$ be the corresponding eigenvector, then following the definition of eigenvector, we can write \[\mathcal{T}_W(W^*)(\psi)=\lambda\psi.\] Now from the equation \ref{alg_eqn_pde3}, we can derive the following equation:
    \begin{subequations}
\begin{eqnarray}
\left(-\Delta -\mathcal{F}_W(W_*)\right)\psi&=&\lambda\left(I-\Delta\right)\psi,~~~~~~~ \mathbf{x}\in \Omega,\\
\frac{\partial \psi(\mathbf{x})}{\partial n}&=&0,~~~~~~~~~~~~~~~~~~~~ \mathbf{x}\in \partial\Omega.
\end{eqnarray}
\end{subequations}
Hence, we can claim that $\lambda$ be an eigenvalue of $\mathcal{T}_W(W^*)$ if $\lambda$ be an eigenvalue of $(1+\mu_j)^{-1}(\mu_jI-\mathcal{F}_W(W_*))$ for all $j \geq 0$. Therefore, $\mathcal{T}_W(W^*)$ is invertible if and only if the matrix $(\mu_jI-\mathcal{F}_W(W_*))$ is invertible, for $j \geq 0$, which further implies $\text{Det}(\mu_jI-\mathcal{F}_W(W_*))\neq 0$. Let us define $\mathcal{H}(\mu_j,d):=\text{Det}(\mu_j I-\mathcal{F}_W(W_*))$, then \[\mathcal{H}(\mu_j,d)= \mu_j^2-\left\{r\left(1-\frac{2u_*}{X_0}\right)-(f(u_*)-1)\left\{\frac{1}{d}+\frac{f'(u_*)}{k_3}\right\}\right\}\mu_j+\frac{(1-f(u_*))G(u_*)}{d}.\] Note that, $\mathcal{H}(\mu,d)=0$ is a quadratic equation of $\mu$ and two real roots $\mu^{\pm}$ of the equation will be obtained if the discriminant of this quadratic equation be positive (\ref{localstablespace}). Then the set $\mathcal{S}(d):=\left(\mu^-,\mu^+\right)-\mathbb{R}_{<0 }$ be well defined. Now, to compute $\mbox{index}(\mathcal{T}(.), W_*)$, we will use the following lemma as given in \cite{Pang_Wang_2003}.

\noindent\textbf{Lemma}\label{lemma} Suppose $\mathcal{H}(\mu_j,d)\neq 0$ for all $j \geq 0$, then \[\mbox{index}(\mathcal{T}(.), W_*)=(-1)^\xi,\] where $\xi=\underset{\mu_j\in \mathcal{S}(d)\cap \mathcal{S}_{\Delta}}{\sum}m(\mu_j)$ if the set $\mathcal{S}(d)\cap \mathcal{S}_{\Delta}$ is non-empty, where $m(\mu_j)$ denotes the algebraic multiplicity of the eigenvalue $\mu_j$, and if the set $\mathcal{S}(d)\cap \mathcal{S}_{\Delta}$ is empty then $\xi=0$.  In particular, if $\mathcal{H}(\mu_j,d)>0$ for all $j \geq 0$, then $\xi=0$.

We now present a theorem that ensures the existence of at least one positive heterogeneous solution. For this, we make the following assumption on the prey density of homogeneous steady-state
\begin{equation}\label{inequality}
    \left(d\left[G(u_*)+\frac{f'(u_*)f(u_*)}{k_3}\right]+1-f(u_*)\right)>2\sqrt{d\left(1-f(u_*)\right)G(u_*)}.
\end{equation}

\begin{thm} Assume that the prey density of a coexisting homogeneous steady-state of the system (\ref{alg_eqn_pde9}) satisfies the condition (\ref{inequality}). If there exists some integers $0\leq n_1<n_2$ such that $\mu^- \in (\mu_{n_1},\mu_{n_1+1})$, $\mu^+ \in (\mu_{n_2},\mu_{n_2+1})$ and $\sum_{j=n_1+1}^{n_2}m(\mu_{j})$ is an odd number, then the system (\ref{alg_eqn_pde2}) has at least one non-constant positive solution.
\end{thm}

\begin{proof} We proof this theorem by applying Leray-Schauder degree theory \cite{degree} with homotopy invariance. Let us denote the coexisting homogeneous steady-state by $W_*=(u_*,v_*)$, where $u_*$ satisfies the condition (\ref{inequality}). Suppose that the system (\ref{alg_eqn_pde2}) does not admit any non-constant positive solution whenever $d<d_*$, where $d_*$ is defined in Theorem \ref{non-existence}. Note that
\begin{itemize}
    \item[(i)] the system (\ref{alg_eqn_pde2})  does not exhibit non-constant positive solution with diffusion coefficient $d_*$,
    \item[(ii)] $\mathcal{H}(\mu_j,d_*)> 0$ for all $j \geq 0$.
\end{itemize}
Now for $\tau\in[0,1]$, let us define 
$\overline{\mathcal{F}}(\tau, W)=(F_1,(\tau d+(1-\tau) d_*)^{-1}F_2)^T$ and consider the problem
\begin{subequations}\label{alg_eqn_pde4}
\begin{eqnarray}
-\Delta W(\mathbf{x})&=&\overline{\mathcal{F}}(\tau, W),\\
\frac{\partial W(\mathbf{x})}{\partial n}&=&0.
\end{eqnarray}
\end{subequations}
The system (\ref{alg_eqn_pde4}) has positive non-constant solution \textbf{W} if and only if it satisfies the following equation
\[\overline{\mathcal{T}}(\tau,\textbf{W}):=\textbf{W}-(I-\Delta)^{-1}\left[\textbf{W}+\overline{\mathcal{F}}(\tau,\textbf{W})\right]=0~\mbox{for}~\textbf{W}\in \mathcal{U}^+.\]
A straightforward calculation yields
\[\overline{\mathcal{T}}_W(\tau,W^*):=I-(I-\Delta)^{-1}\left[I+\overline{\mathcal{F}}_W(\tau,W_*)\right],\]
where $\overline{\mathcal{F}}_W(\tau,W_*)=\left[
    \begin{array}{ccc}
    r(1-\frac{2u_*}{X_)})-f'(u_*)\frac{f(u_*)-1}{k_3} & -f(u_*) \\
    f'(u_*)\frac{f(u_*)-1}{(\tau d+(1-\tau) d_*)k_3} & \frac{1-f(u_*)}{\tau d+(1-\tau) d_*}
    \end{array}\right]$. Therefore \textbf{W} is a positive non-constant solution of the system (\ref{alg_eqn_pde2}) if and only if it is a solution to the system (\ref{alg_eqn_pde4}), when $\tau=1$. Observe that $\overline{\mathcal{F}}(1,\textbf{W})=\mathcal{F}(\textbf{W})$ and $\overline{\mathcal{T}}(1,\textbf{W})=\mathcal{T}(\textbf{W})$. If the condition (\ref{inequality}) holds, then from the quadratic equation $\mathcal{H}(\mu,d)=0$, we obtain two real roots $\mu^\pm$ such that $0<\mu^-<\mu^+$. Let us choose two natural number $0\leq n_1<n_2$ in such a way that $\mu^- \in (\mu_{n_1},\mu_{n_1+1})$ and $\mu^+ \in (\mu_{n_2},\mu_{n_2+1})$. Observe that 
    $\mathcal{H}(\mu_j,d)< 0 ~\mbox{whenever}~n_1+1\leq j\leq n_2~\mbox{and}~\mathcal{H}(\mu_j,d)>0~\mbox{otherwise}$. Now, using the Lemma, we can further derive that
    \[\xi=\underset{\mu_j\in \mathcal{S}(d)\cap \mathcal{S}_{\Delta}}{\sum}m(\mu_j)=\overset{n_2}{\underset{j=n_1+1}{\sum}}m(\mu_j),\] 
    which is an odd number, according to our assumption. Therefore 
    \[\mbox{index}(\overline{\mathcal{T}}(1,.), W_*)=(-1)^\xi=-1.\]
    Conversely, we know $\mathcal{H}(\mu_j,d_*)>0$ for all $j\geq 0$, which implies
    \[\mbox{index}(\overline{\mathcal{T}}(0,.), W_*)=(-1)^\xi=(-1)^0=1.\]
    Using Theorem \ref{Boundss}, we can find constants $\underline{C}$ and $\overline{C}$ such that, for all $\tau \in [0,1]$, the positive solution $W=(u,v)$ satisfies the condition  $\underline{C}<u,v<\overline{C}$. Consider the set $\mathbb{B}=\{W\in \mathcal{W}:\underline{C}<u,v<\overline{C}\}$. $W_*$ being unique constant solution of the system (\ref{alg_eqn_pde4}) satisfying the condition (\ref{inequality}), we have $W_*\in \mathbb{B}$. From the homotopy invariance of Leray-Schauder degree, we deduce 
    \[\deg(\overline{\mathcal{T}}(0,.),0,\mathbb{B})=\deg(\overline{\mathcal{T}}(1,.),0,\mathbb{B}).\]
    Also both the equations $\overline{\mathcal{T}}(0,\textbf{W})=0$ and $\overline{\mathcal{T}}(1,\textbf{W})=0$ has unique solution $W_*$ in $\mathbb{B}$. Now check that $\deg(\overline{\mathcal{T}}(0,.),0,\mathbb{B})=\mbox{index}(\overline{\mathcal{T}}(0,.), W_*)=1$ and $\deg(\overline{\mathcal{T}}(1,.),0,\mathbb{B})=\mbox{index}(\overline{\mathcal{T}}(1,.), W_*)=-1$, which is a contradiction. Therefore, the system has at least one non-constant positive solution.
    \end{proof}
    
\subsection{Turing bifurcation} 
We have already discussed the conditions for the existence and non-existence of heterogeneous steady states. However, the underlying mechanism leading to the existence of a non-constant steady state remains to be examined. To explore this, we analyze the diffusion-driven instability conditions under which a homogeneous steady state, stable under small-amplitude homogeneous perturbations, becomes unstable due to small-amplitude heterogeneous perturbations—a phenomenon known as Turing instability. To determine the condition for Turing instability, we introduce spatial perturbations around the homogeneous steady state by setting $u=u_*+\epsilon \bar{u}e^{(\lambda t+\iota k \mathbf{x})}$ and $v=v_* + \epsilon \bar{v}e^{(\lambda t+\iota k \mathbf{x})}$, where $|\epsilon| \ll 1$ and $\bar{u} ~\mbox{and}~ \bar{v}$ be a non-trivial solution of the following system:
\[M_k\left[
    \begin{array}{ccc}
    \bar{u} \\
    \bar{v}
    \end{array}\right]=\left[
    \begin{array}{ccc}
    r\left(1-\frac{2u_*}{X_0}\right)-\frac{f'(u_*)\left(f(u_*)-1\right)}{k_3}-k^2-\lambda & -f(u_*) \\
    \frac{f'(u_*)\left(f(u_*)-1\right)}{k_3} & 1-f(u_*)-d k^2-\lambda
    \end{array}\right]\left[
    \begin{array}{ccc}
    \bar{u} \\
    \bar{v}
    \end{array}\right]=\left[
    \begin{array}{ccc}
    0 \\
    0
    \end{array}\right].\]
Now the above system posses a non-trivial solution if $\det(M_k)=0$, which leads to the characteristic equation
\[\lambda^2-\Phi_k\lambda+\Psi_k=0,\]
where \[\Phi_k=G(u_*)+\frac{f'(u_*)f(u_*)}{k_3}+1-f(u_*)-(1+d)k^2,~ \mbox{and}\]\[\Psi_k=dk^4-k^2\left(d\left[G(u_*)+\frac{f'(u_*)f(u_*)}{k_3}\right]-(f(u_*)-1)\right)+\left(1-f(u_*)\right)G(u_*).\] 
Since the homogeneous steady state is stable under small-amplitude homogeneous perturbations, we can conclude that $\Phi_k<0,~\forall~k$. If $\Psi_k>0,~\forall~k$, then the homogeneous steady state remains stable under small amplitude heterogeneous perturbations. Therefore, we seek a condition under which $\Psi_k<0$ for some $k$, as this would indicate the onset of Turing instability. The minimum value of $\Psi_k$ is \[\underset{k>0}{\min} \Psi_k=\left(1-f(u_*)\right)G(u_*)-\frac{\left(d\left[G(u_*)+\frac{f'(u_*)f(u_*)}{k_3}\right]-(f(u_*)-1)\right)^2}{4d}.\]
For $\Psi_k$ to transition from positive to negative for some $k$, we must satisfy the condition $\underset{k>0}{\min} \Psi_k=0$, which defines the Turing bifurcation threshold. The corresponding critical value of $d$ is given by

\[d_c=\frac{(f(u_*)-1)\left[\left(3G(u_*)+\frac{f(u_*)f'(u_*)}{k_3}\right)+2\sqrt{G(u_*)\left(2G(u_*)+\frac{f(u_*)f'(u_*)}{k_3}\right)}\right]}{\left(G(u_*)+\frac{f(u_*)f'(u_*)}{k_3}\right)^2}.\]
Under a certain choice of parameter values, if the critical value $d_c$ is positive, then for $d>d_c$, the coexisting homogeneous steady state becomes unstable, and the system can exhibit a heterogeneous steady state. However, the existence of a non-constant steady state must be verified through numerical simulations and cross-checked against the analytical existence conditions for non-constant steady states.

\section{Structural Sensitivity: Numerical illustration}{\label{NumRes}}

The analytical criteria for the occurrence of saddle-node, Hopf, Bogdanov–Takens, and cusp bifurcations around the homogeneous steady state have been discussed in previous sections, along with the derivation of the condition for Turing instability. These analytical results are based on the assumption that the functional responses, though not explicitly specified, share several similar qualitative properties. We now proceed to numerically analyze these local bifurcations and their impact on the global dynamics of the temporal model, as well as on spatio-temporal pattern formation in the spatially extended model. For this purpose, we consider two explicitly defined parametrizations of functional responses: the Holling Type II response ($f_H(u) = {k_{1_H} u}/{1 + k_{2_H} u}$) and the Ivlev response ($f_I(u) = k_{1_I}(1 - \exp(-k_{2_I} u))$). We begin by verifying and comparing the local and global dynamics of the temporal model \eqref{algeb_equn} under each functional response. Then, we examine how the different parametrizations influence the Turing instability criteria and the resulting spatial patterns.
 
\begin{figure}[ht]
\centering
\mbox{\subfigure[]{\includegraphics[width=8.7cm]{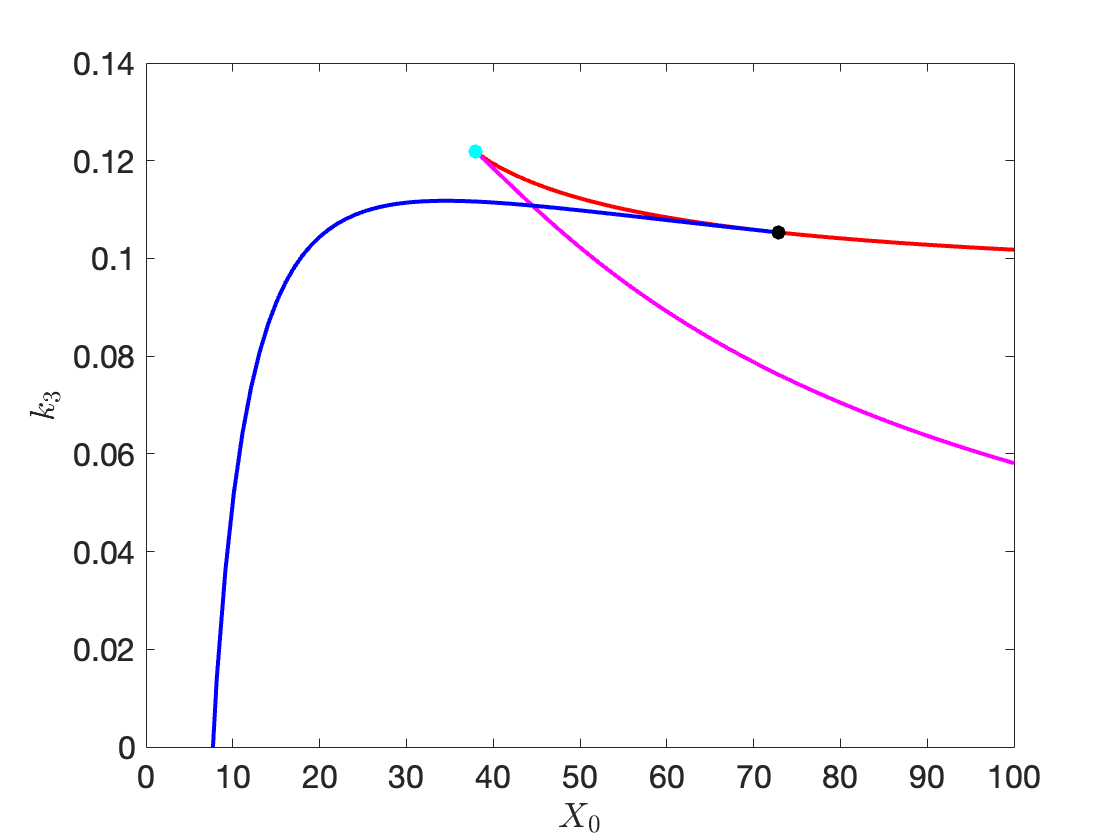}}
\subfigure[]{\includegraphics[width=8.7cm]{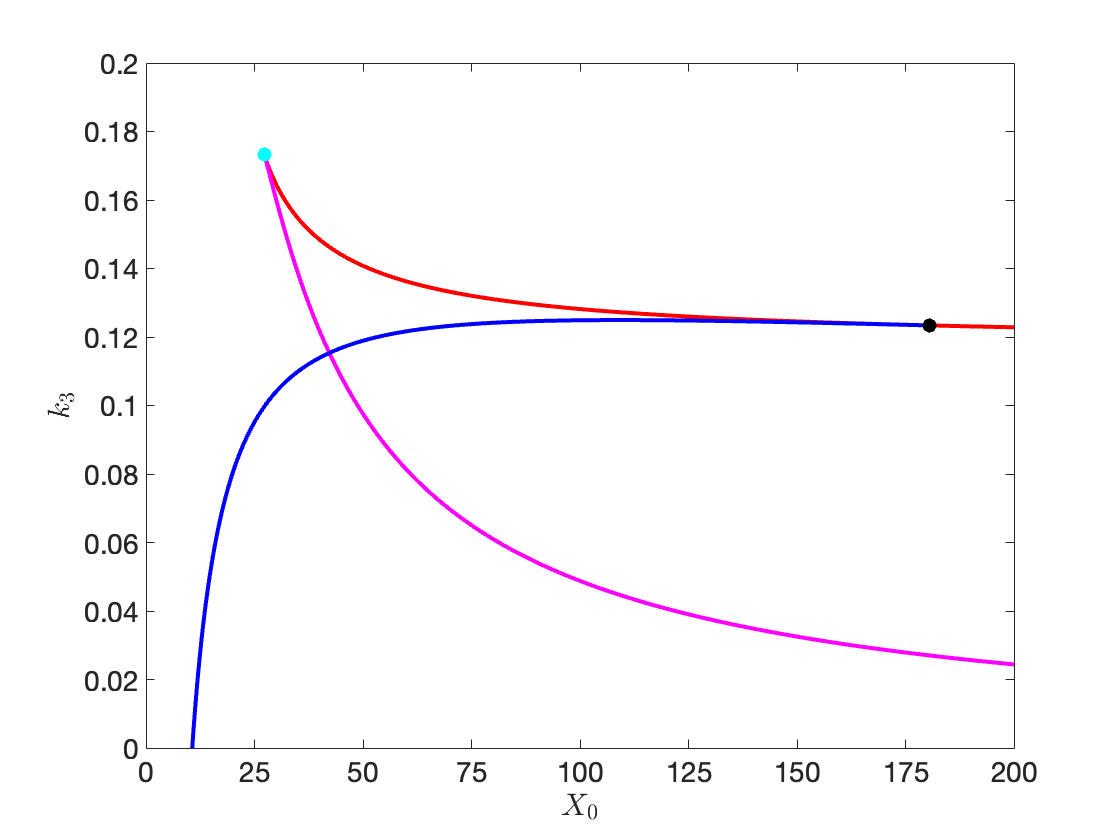}}}
\caption{(a) 2d Bifurcation for Holling type II (b) 2d Bifurcation for Ivlev, with $X_0$ and $k_3$ as bifurcation parameters. Here $X_0$ is varying along $x$-axis and $k_3$ is varying along $y$-axis.}
\label{Fig:Fig1}
\end{figure}

To initiate our numerical study, we selected the parameter values for the system \eqref{algeb_equn} with the Holling Type II functional response as outlined in \cite{IJBC}. The corresponding parameters for the Ivlev functional response were estimated using the least-squares approximation method, with the Holling Type II response serving as the baseline. Since our objective is to investigate the structural sensitivity of spatial patterns, the approximation of the Ivlev functional response was carried out in such a way that the homogeneous steady state—satisfying the Turing instability condition—remains closely aligned between both models. Table \ref{Table:1} lists the fixed parameter values associated with the Holling Type II functional response and the estimated parameters for the Ivlev functional response.

\begin{table}[h]
\begin{center}
\begin{tabular}  {|c|c|} \hline
Parameter  &  Values \\ 
\hline 
$k_{1_H}$ & 1.3 \\ 
\hline 
$k_{2_H}$ & 1 \\
\hline 
$k_{1_I}$ & 1.2367 \\ 
\hline 
$k_{2_I}$ & 0.4101 \\ 
\hline 
\end{tabular}
\caption{Fixed parameter values corresponding to the two considered functional responses.}
\label{Table:1}
\end{center}
\end{table}

\begin{figure}[ht]
\centering
\includegraphics[width=14cm]{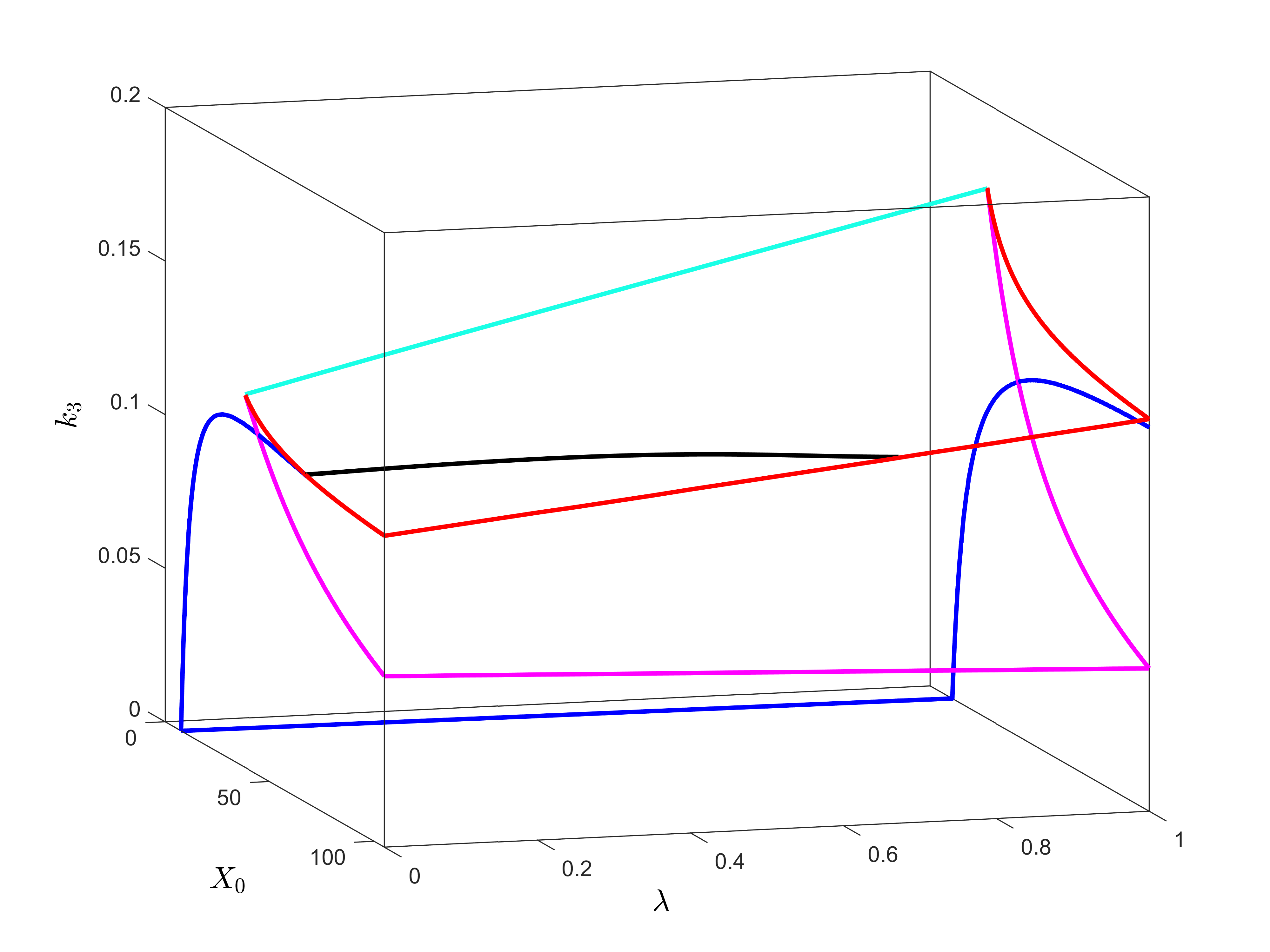}
\caption{Three parametric bifurcation diagram for the system \eqref{algeb_equn} for the fixed parameter values $r=0.24$, $b_{H}=1.3$, $a_{H}=1$, $b_{I}=1.2367$, $a_{I}=0.4101$. The surface edged by the blue curves represents the Hopf bifurcation, the surface enclosed by the red and magenta lines denote two saddle-node bifurcation. The cyan curve representing cusp bifurcation, and the black curve indicates the Bogdanov–Takens bifurcation. }
\label{Fig:Fig2}
\end{figure}

Analytical results discussed in Section \ref{LocBif} reveal the possibility of different bifurcation structures depending on the choice of parameter values $r$, $X_0$, and $k_3$. Here, we fix the parameter value $r = 0.24$ and treat $X_0$ and $k_3$ as bifurcation parameters to construct two-parameter bifurcation diagrams for the models with Holling Type II and Ivlev functional responses, shown in Figures \ref{Fig:Fig1}(a) and \ref{Fig:Fig1}(b), respectively. In both cases, we identify two saddle-node bifurcation curves that intersect at a cusp bifurcation point, as well as a Hopf bifurcation curve that intersects the upper branch of the saddle-node bifurcation curve at a Bogdanov–Takens bifurcation point. However, notable differences arise in the threshold values of these bifurcations. In the model with the Holling Type II functional response, the two saddle-node bifurcation curves lie much closer together compared to those in the Ivlev-based model. Moreover, while the Hopf bifurcation curve increases monotonically in the Ivlev model, it exhibits a non-monotonic trend in the Holling Type II model—initially increasing and then decreasing. As a result, the distance between the cusp and Bogdanov–Takens bifurcation points is larger in the Ivlev model than in the Holling Type II model. The coordinates of the Bogdanov–Takens bifurcation point are $(72.9180, 0.1053)$ for the Holling Type II model and $(180.5614, 0.1233)$ for the Ivlev model, while the coordinates of the cusp bifurcation point are $(38.0003, 0.1217)$ and $(27.2783, 0.1732)$ for the Holling Type II and Ivlev models, respectively.

 \begin{figure}[ht]
\centering
\mbox{\subfigure[]{\includegraphics[width=8.7cm]{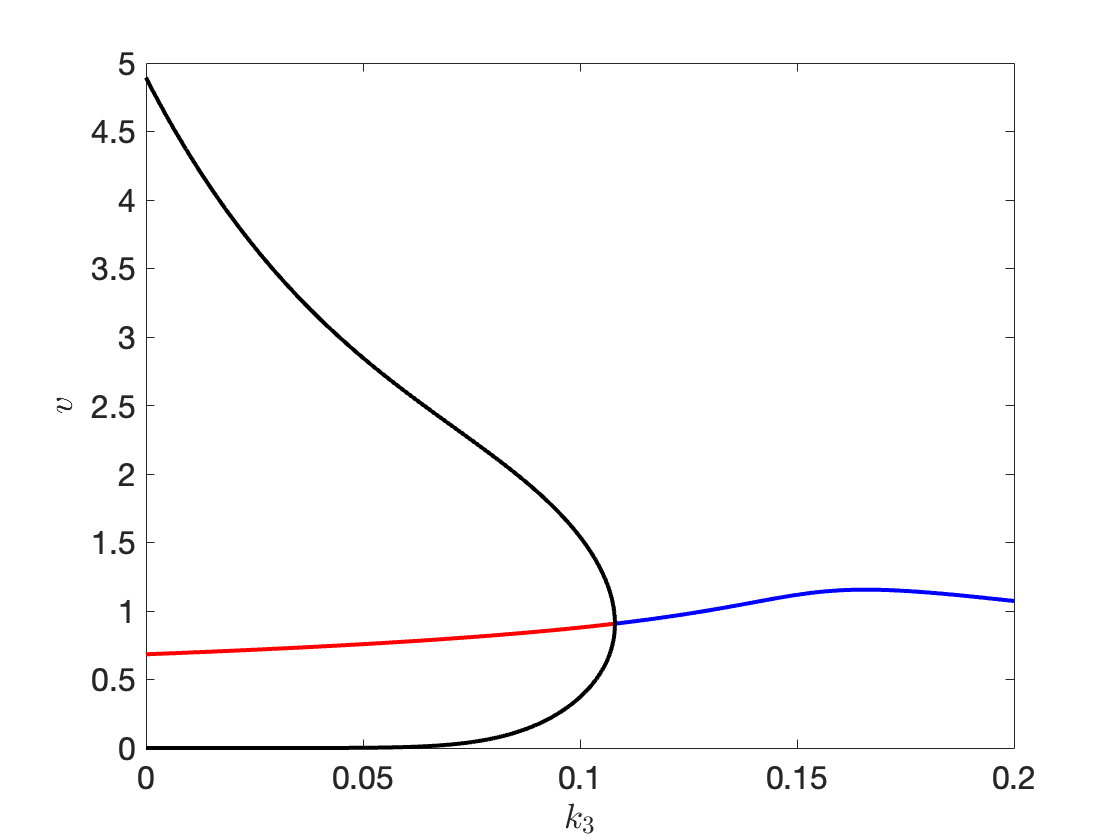}}
\subfigure[]{\includegraphics[width=8.7cm]{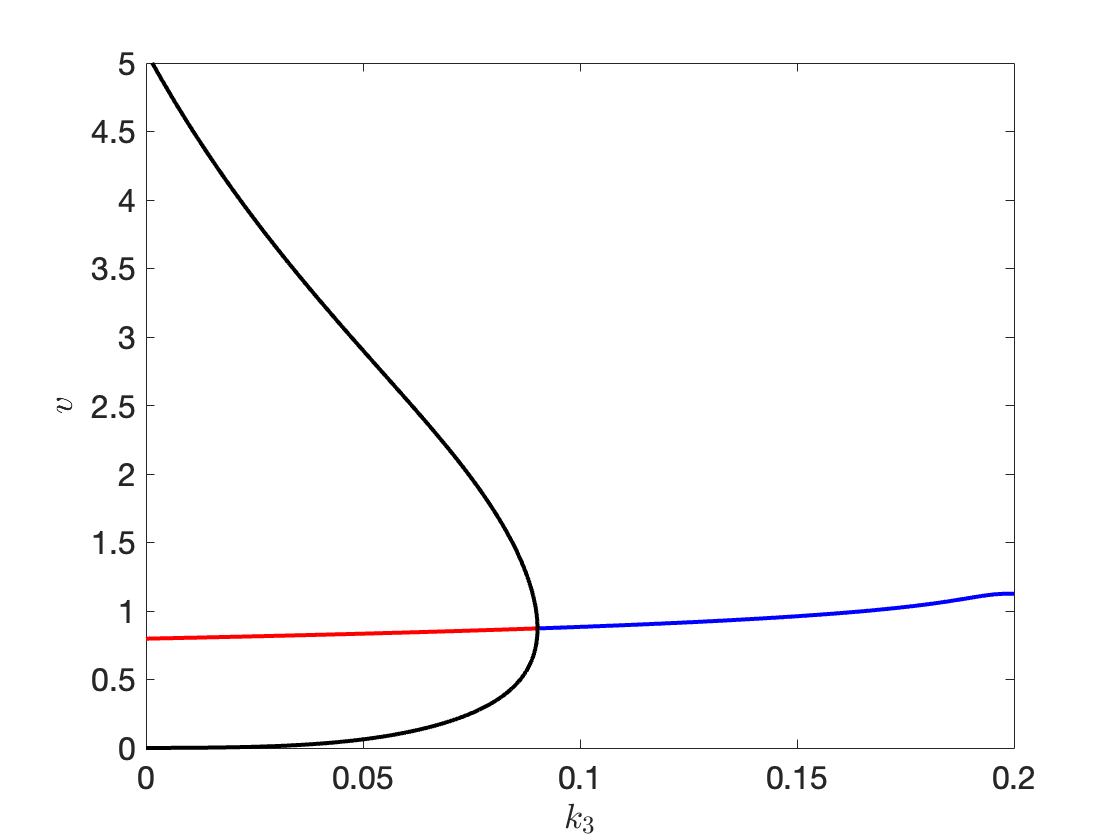}}}
\caption{One parametric bifurcation diagram for the system \eqref{algeb_equn} with (a) the Holling type II functional response with parameter values corresponding to functional response are $k_{1_H}=1.3, k_{2_H}=1$, (b) the Ivlev functional response with parameter values corresponding to functional response are $k_{1_I}=1.2367, k_{2_I}=0.4101$. The fixed parameter values are $X_0=23~\mbox{\&}~ r=0.24$. $k_3$ is the bifurcation parameter}
\label{Fig:Fig3}
\end{figure}

To investigate whether any intermediate bifurcation structures emerge as we transition from the Holling Type II functional response to the Ivlev functional response, we consider a mixed functional response defined as:  \[f(u)=\sigma f_I(u)+(1-\sigma) f_H(u),~\sigma \in [0,1],\] which represents a convex combination of the Holling Type II and Ivlev functional responses. For further details, see \cite{Aldebert2019}. We treat $\sigma$ as an additional bifurcation parameter alongside $X_0$ and $k_3$, and generate the corresponding bifurcation diagram shown in Fig. \ref{Fig:Fig2}. This three-dimensional bifurcation diagram confirms that the shifts in bifurcation thresholds between the two limiting cases—$\sigma = 0$ (corresponding to the Holling Type II model) and $\sigma = 1$ (corresponding to the Ivlev model)—are a result of changes in the parametrization of the functional response.

A study of the one-parameter bifurcation diagram reveals significant alterations in the global bifurcation structure. From the two-parameter bifurcation diagrams for both models, shown in Fig.~\ref{Fig:Fig1}, it is evident that when $X_0 = 23$ and the parameter $k_3$ is varied, the system does not undergo a saddle-node bifurcation. Furthermore, Theorem \ref{thm1} confirms that under these parameter constraints, both models possess a unique coexisting steady state for different values of $k_3$. In both cases, this steady state undergoes a supercritical Hopf bifurcation, although a notable shift in the bifurcation threshold is observed. Specifically, the Hopf bifurcation occurs at $k_3 = 0.108$ for the model with the Holling Type II functional response, and at $k_3 = 0.0902$ for the model with the Ivlev functional response. The bifurcation diagrams for both models appear topologically equivalent, with the stable limit cycle emerging from the supercritical Hopf bifurcation exhibiting nearly identical oscillation amplitudes. Additionally, for small values of $k_3$, large-amplitude oscillations are observed in both models.

\begin{figure}[ht]
\centering
\mbox{\subfigure[]{\includegraphics[width=8.7cm]{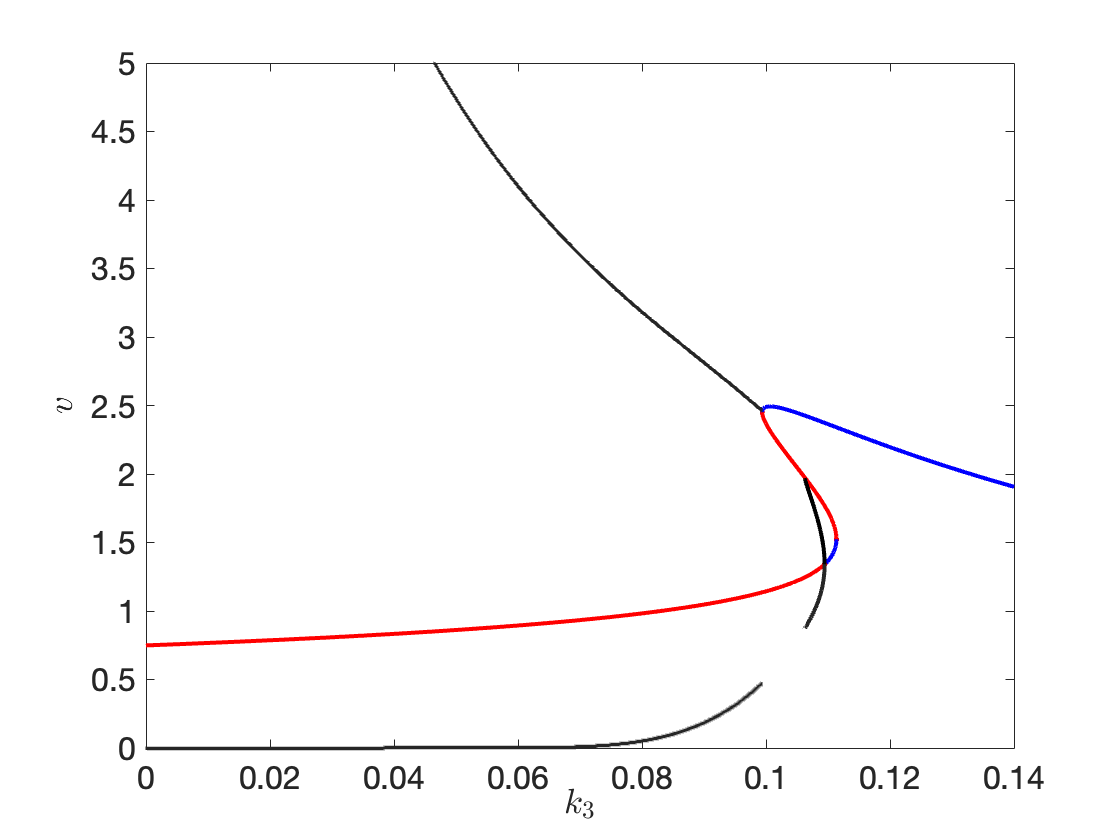}}
\subfigure[]{\includegraphics[width=8.7cm]{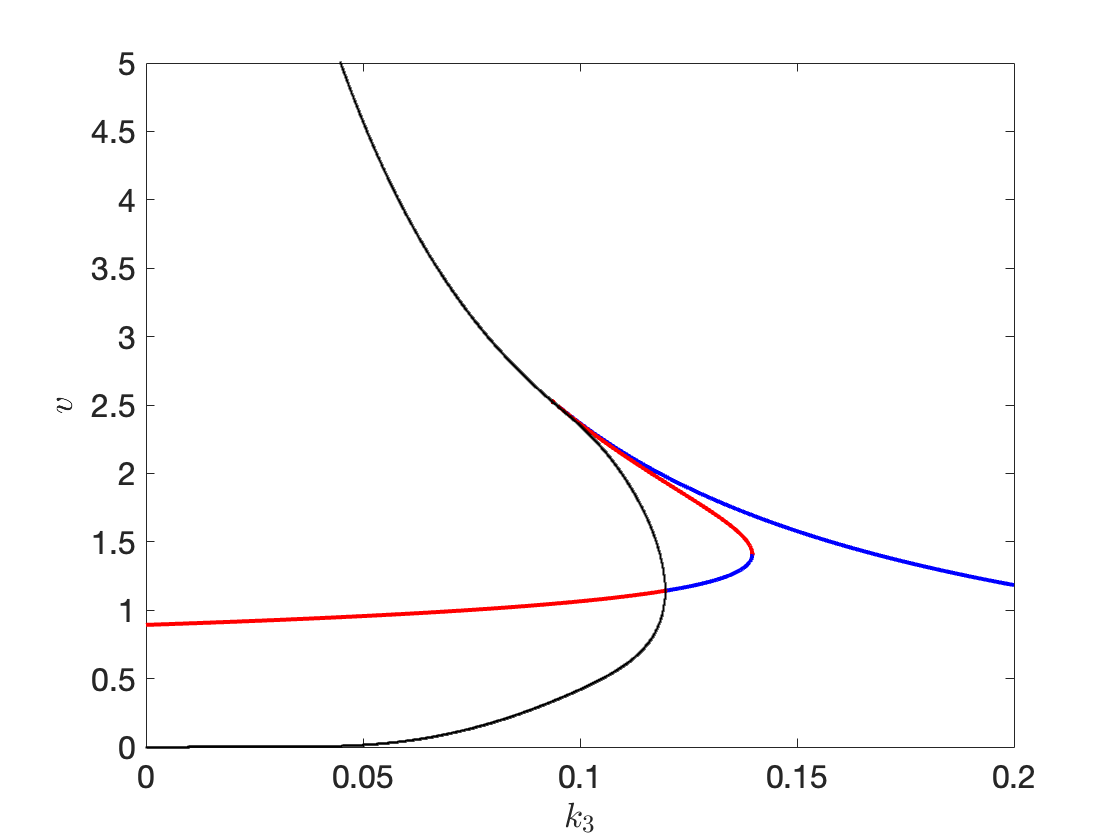}}}
\caption{One parametric bifurcation diagram for the system \eqref{algeb_equn} with (a) the Holling type II functional response with parameter values corresponding to functional response are $k_{1_H}=1.3, k_{2_H}=1$, (b) the Ivlev functional response with parameter values corresponding to functional response are $k_{1_I}=1.2367, k_{2_I}=0.4101$. The fixed parameter values are $X_0=52~\mbox{\&}~ r=0.24$. $k_3$ is the bifurcation parameter}
\label{Fig:Fig4}
\end{figure}

Keeping all other parameter values unchanged and setting $X_0 = 52$, Theorem \ref{thm1} confirms that both systems exhibit at least one and at most three coexisting steady states as the parameter $k_3$ varies. For the model with the Holling Type II functional response, three steady states coexist within the range $k_3 \in [0.1, 0.115]$, whereas for the model with the Ivlev functional response, this range extends to $k_3 \in [0.1, 0.14]$. This bistable scenario arises due to two successive saddle-node bifurcations. For the chosen set of parameter values, the steady state with the highest prey density remains stable throughout its existence, while the steady state with intermediate prey density is always unstable. Meanwhile, the coexisting steady state with the lowest prey density undergoes a supercritical Hopf bifurcation in both models. The Hopf bifurcation threshold occurs at $k_3 = 0.1094$ for the Holling Type II model and at $k_3 = 0.0902$ for the Ivlev model. The stable limit cycle that emerges from this supercritical Hopf bifurcation disappears via a homoclinic bifurcation at $k_3 = 0.1062$ and reappears through another homoclinic bifurcation at $k_3 = 0.0933$ for the Holling Type II model. Within the interval $k_3 \in (0.0933, 0.1062)$, this limit cycle vanishes, and trajectories originating from the interior of the first octant eventually settle into the coexisting steady state with the highest prey density, following transient oscillations. However, this disappearance and reappearance of the limit cycle has not been observed in the Ivlev model under the same parameter conditions.

\begin{figure}[ht]
\centering
\subfigure[]{\includegraphics[width=9.5cm]{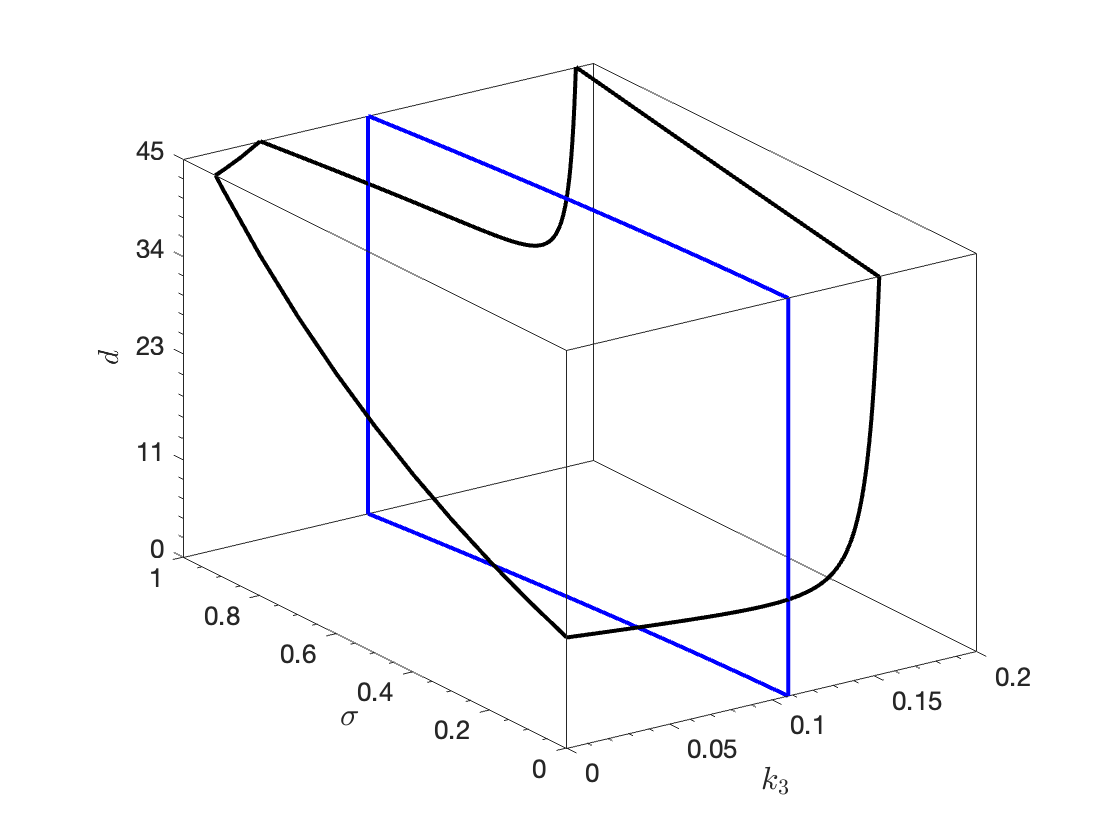}}\\
\subfigure[]{\includegraphics[width=6.5cm]{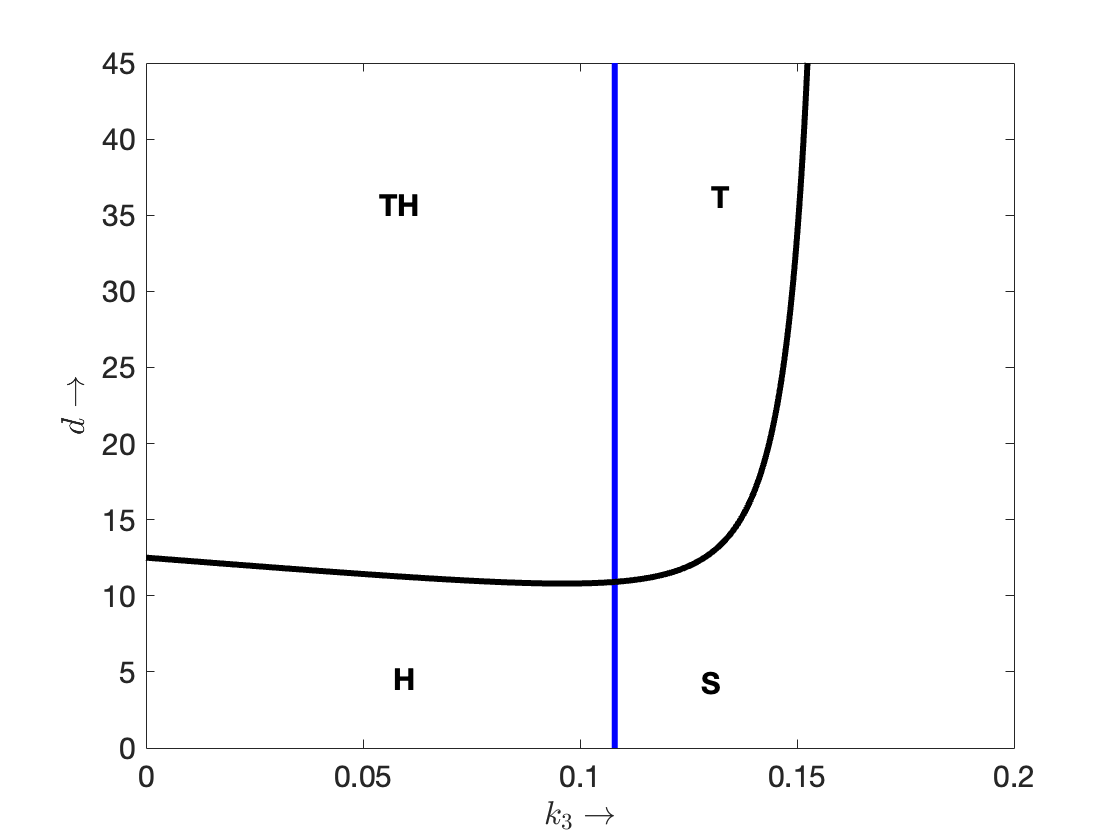}}
\subfigure[]{\includegraphics[width=6.5cm]{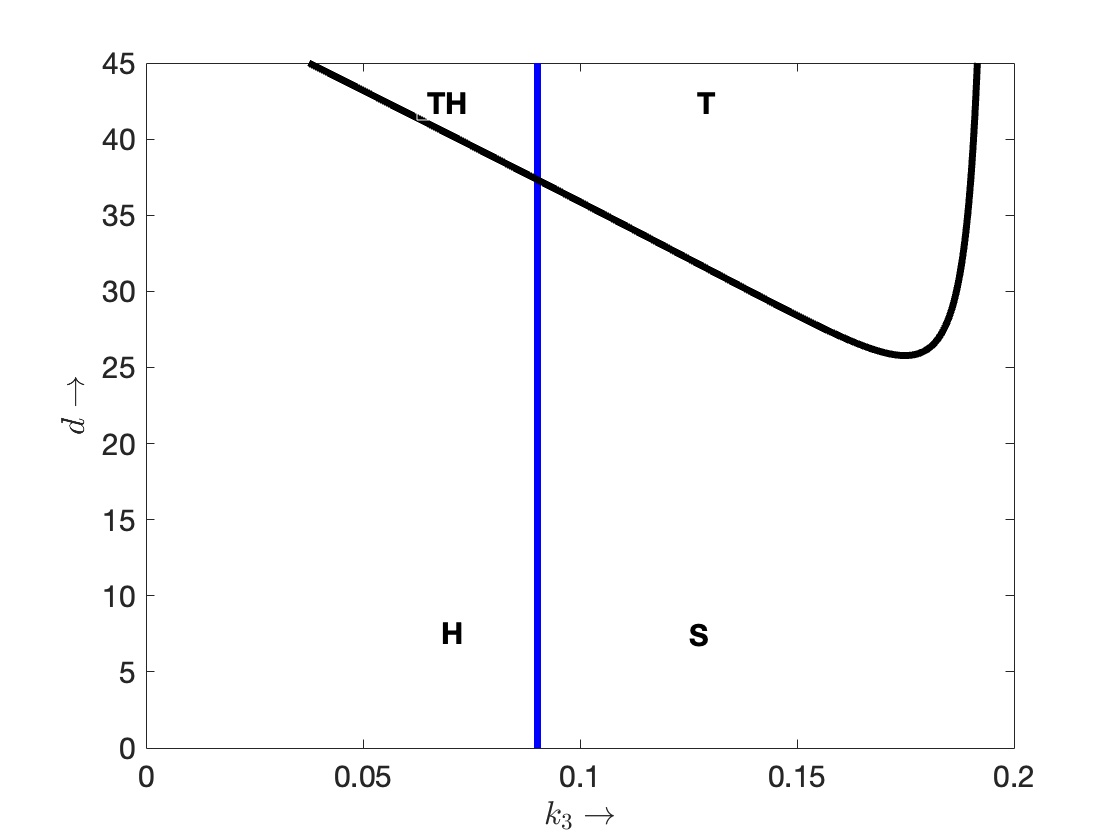}}
\caption{(a) Three parametric plot of Turing and Hopf bifurcation curves for the system \eqref{alg_eqn_pde} with the fixed parameter values $X_0=23$, $r=0.24$, $b_{H}=1.3$, $a_{H}=1$, $b_{I}=1.2367$, $a_{I}=0.4101$, representing a continuous shift of Turing bifurcation threshold due to continuous switch between two parameterizations of functional response: Holling's $\sigma=0$ to Ivlev's $\sigma=1$. (b) The plot of the Turing bifurcation curve (black) and Hopf bifurcation threshold (blue) in $(k_3,d)$-parameter space for the model with Holling type II functional response. (c) The plot of the Turing bifurcation curve (black) and Hopf bifurcation threshold (blue) in $(k_3,d)$-parameter space for the model with Ivlev functional response.}
\label{Fig:Fig5}
\end{figure}

\subsection{Numerical simulation for spatio-temporal model}

Now, we explore the structural sensitivity of the spatially extended model, considering $\Omega \in \mathbb{R}^2$ as a square domain. Following our study of the temporal model with two different parameter setups, we conduct a similar analysis here for the spatial model: one with $X_0 = 23$ and another with $X_0 = 52$, while varying the parameter $k_3$. The parameter $r$ is fixed at $0.24$, and the values associated with the functional responses are the same as those listed in Table~\ref{Table:1}. The objective is to examine the impact of different parametrizations of the functional response on diffusion-driven instability and the resulting spatial patterns.

For $X_0=23$, Fig.~\ref{Fig:Fig5} (b,c) illustrates the region in the $(k_3,d)$ parameter space where Turing and/or Turing-Hopf instabilities take place. The Turing bifurcation curve and the Hopf bifurcation threshold of the system \eqref{algeb_equn} partition the parameter space into four distinct dynamical regions: stable(region S), Turing (region T), Turing-Hopf (region TH) and Hopf (region H) regions. These regions are shown in Fig.~\ref{Fig:Fig5} (b) for the model with Holling type II and in Fig.~\ref{Fig:Fig5} (c) for the model with Ivlev functional response. A key distinction is, despite the topological equivalence of the unique homogeneous steady state in both models as illustrated in Fig.~\ref{Fig:Fig3}, the Turing bifurcation threshold is significantly higher for the model with the Ivlev functional response than for the model with the Holling Type II functional response. To analyze this transition in the Turing bifurcation threshold, we introduce the convex combination of two considered parameterizations as follows: \[f(u)=\sigma f_I(u)+(1-\sigma) f_H(u),~\sigma \in [0,1].\]Using this mixed function as the functional response, we treat $\sigma$ as the bifurcation parameter together with $k_3~\mbox{and}~d$ to obtain a three-dimensional plot of the Turing and Hopf bifurcation curves with a shifting value of $\sigma$ between $0~\mbox{and}~1$ (see Fig.~\ref{Fig:Fig5}). This figure illustrates the continuous shift of the Turing bifurcation threshold between two terminal cases: $\sigma=0$ for the model with Holling Type II functional response and $\sigma=1$ for the model with Ivlev functional response. These findings emphasize the role of the choice of parameterization of functional response in diffusion-driven instability of homogeneous steady state.

\begin{figure}[ht]
\centering
\mbox{\subfigure[]{\includegraphics[width=5.7cm]{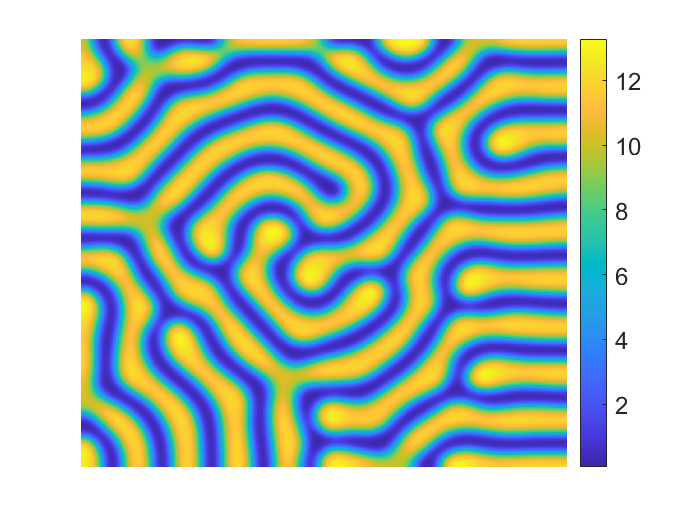}}\subfigure[]{\includegraphics[width=5.7cm]{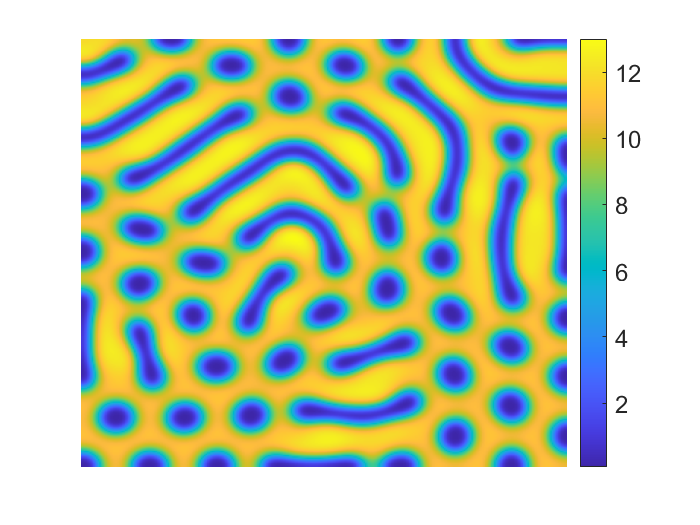}}
\subfigure[]{\includegraphics[width=5.7cm]{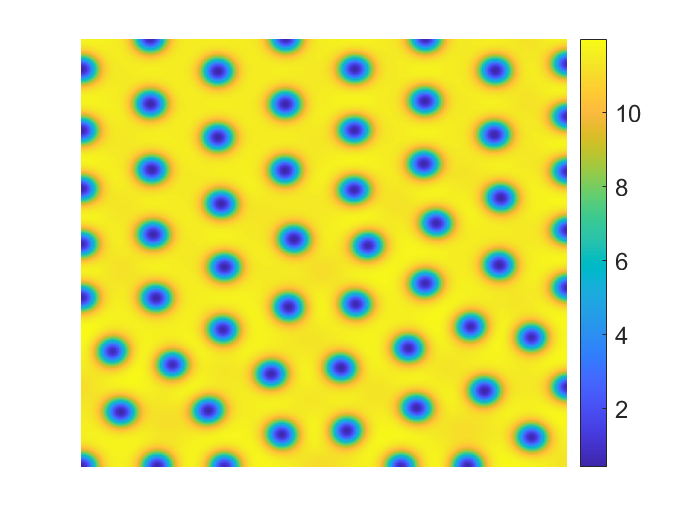}}}
\caption{Patterns that have been observed for the system \eqref{alg_eqn_pde} with Holling type II functional response with the parameter values $k_{1_H}=1.3$, $k_{2_H}=1$, $X_0=23$, $r=0.24$ and for different values of $k_3$. To encapsulate different scenarios distinctly, we have attached patterns observed for the parameter value (a) $k_3=0.04$ (labyrinthine pattern), (b) $k_3=0.1$ (mixture of coldspot and labyrinthine pattern), (c) $k_3=0.15$ (coldspot pattern). }
\label{Fig:Fig6}
\end{figure}

 \begin{figure}[ht]
\centering
\mbox{\subfigure[]{\includegraphics[width=5.7cm]{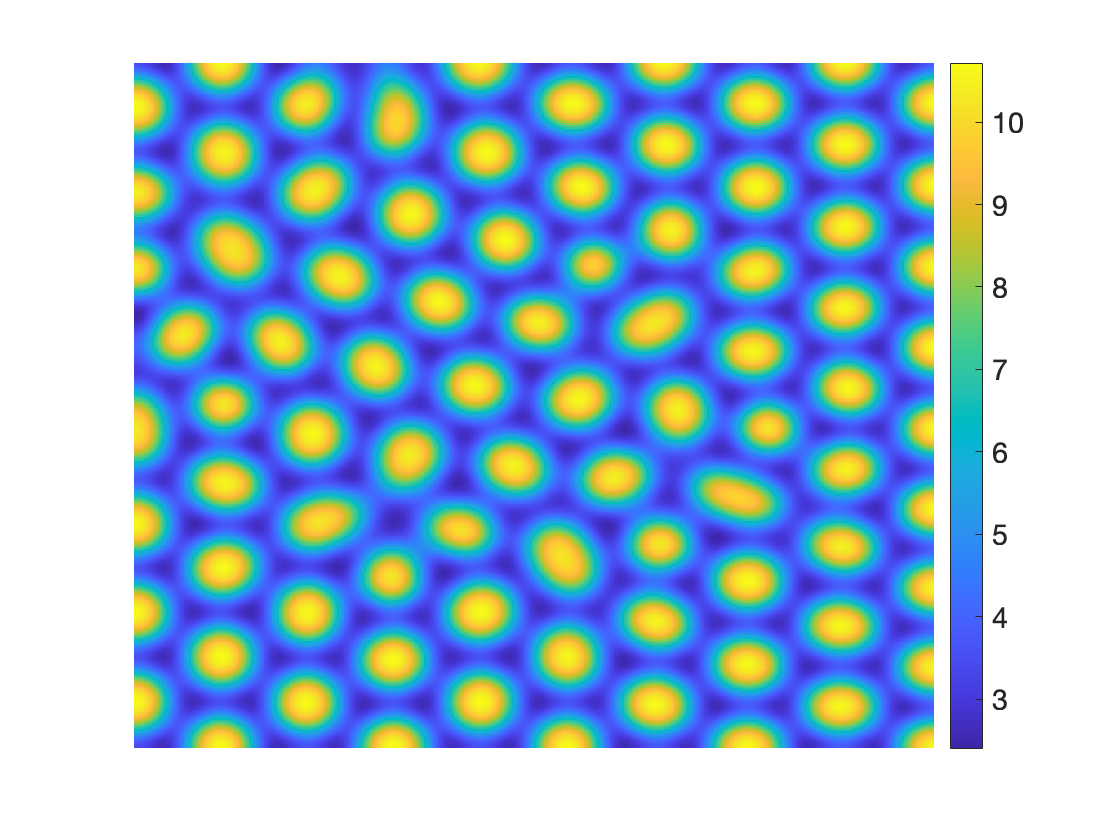}}
\subfigure[]{\includegraphics[width=5.7cm]{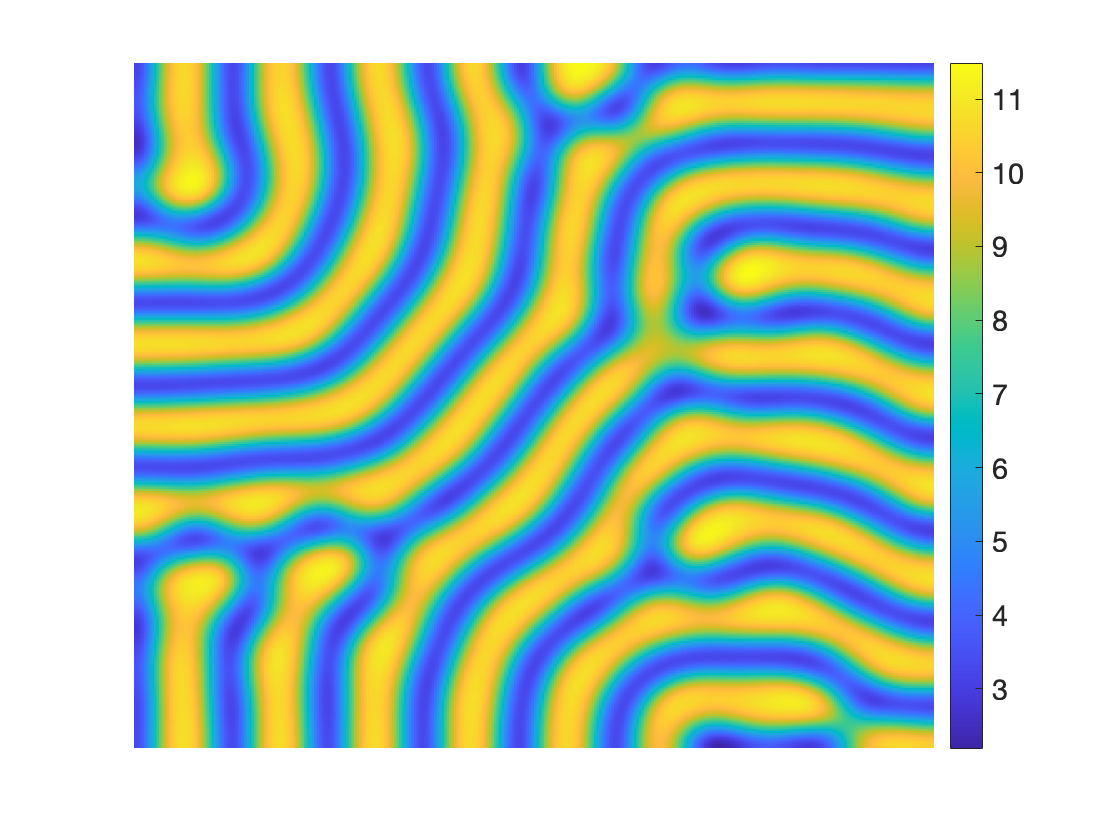}}
\subfigure[]{\includegraphics[width=5.7cm]{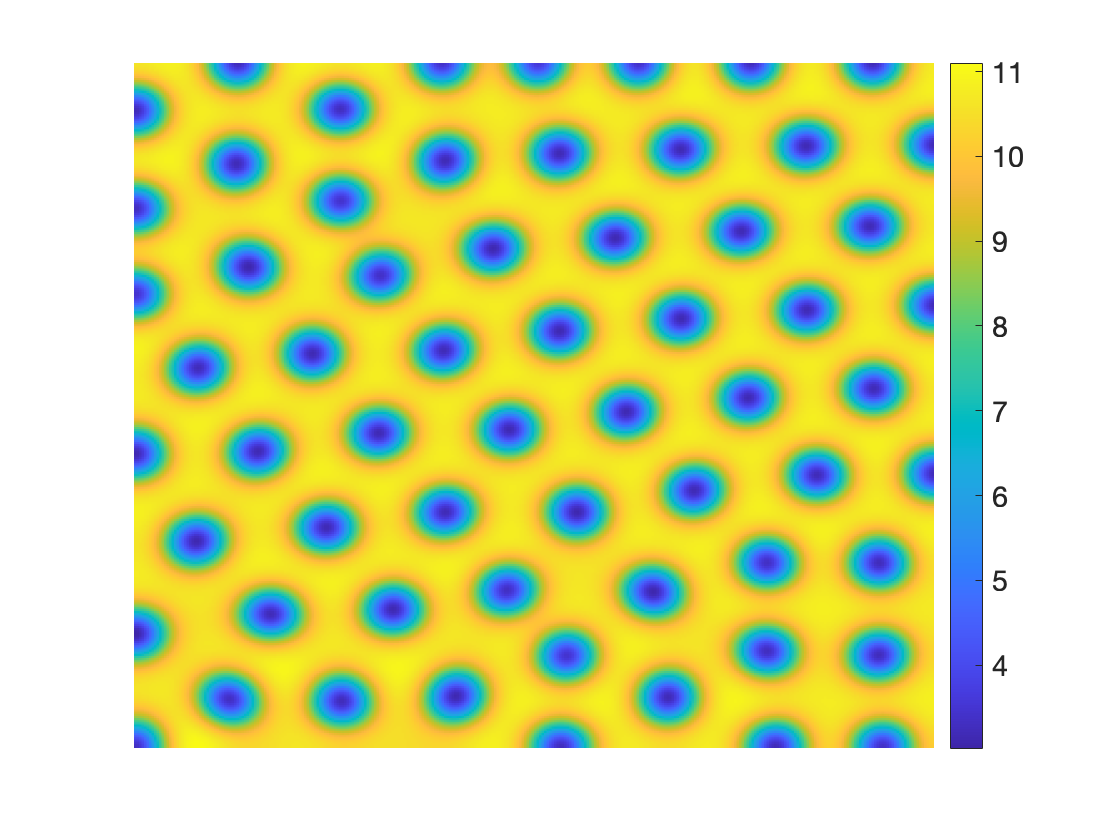}}}
\caption{Patterns that have been observed for the system \eqref{alg_eqn_pde} with Ivlev functional response with the parameter values $k_{1_H}=1.2367$, $k_{2_H}=0.4101$, $X_0=23$, $r=0.24$ and for different values of $k_3$. To encapsulate different scenarios distinctly, we have shown patterns observed for the parameter value (a) $k_3=0.083$ (hotspot pattern), (b) $k_3=0.14$ (labyrinthine pattern), (c) $k_3=0.18$ (coldspot pattern).}
\label{Fig:Fig7}
\end{figure}

Since the Turing bifurcation threshold is higher in the model with the Ivlev functional response, we fix the predator-to-prey self-diffusion coefficient ratio at $d = 45$. This choice ensures that spatial patterns emerge in both models over a wider range of the parameter $k_3$. Notably, in the Ivlev-based model, the Turing instability persists across a broader range of $k_3$ values compared to the model with the Holling Type II functional response. However, the Turing-Hopf instability spans a broader $k_3$ range in the Holling Type II model. Rather than directly comparing the patterns at identical $k_3$ values for both models, we analyze the variety of stationary patterns observed in each case. The possible patterns obtained for the model with the Holling Type II functional response are presented in Fig. \ref{Fig:Fig6}, and those for the model with the Ivlev functional response are shown in Fig. \ref{Fig:Fig7}. These patterns were generated by selecting various $k_3$ values from the Turing and Turing-Hopf regions—specified in the figure captions—and applying a small heterogeneous perturbation to the Turing-unstable homogeneous steady state. Although both models exhibit dynamic patterns for certain values of $k_3$, our focus here is limited to stationary patterns. A notable distinction emerges: the model with the Holling Type II functional response does not produce the hotspot pattern that appears prominently in the Ivlev model (see Fig.~\ref{Fig:Fig7}(a)). This finding underscores how the modeling of predator efficiency, aggregation, and resource exploitation via different functional responses plays a crucial role in generating spatial heterogeneity in ecological systems.

\begin{figure}[ht]
\centering
\mbox{\subfigure[]{\includegraphics[width=8.4cm]{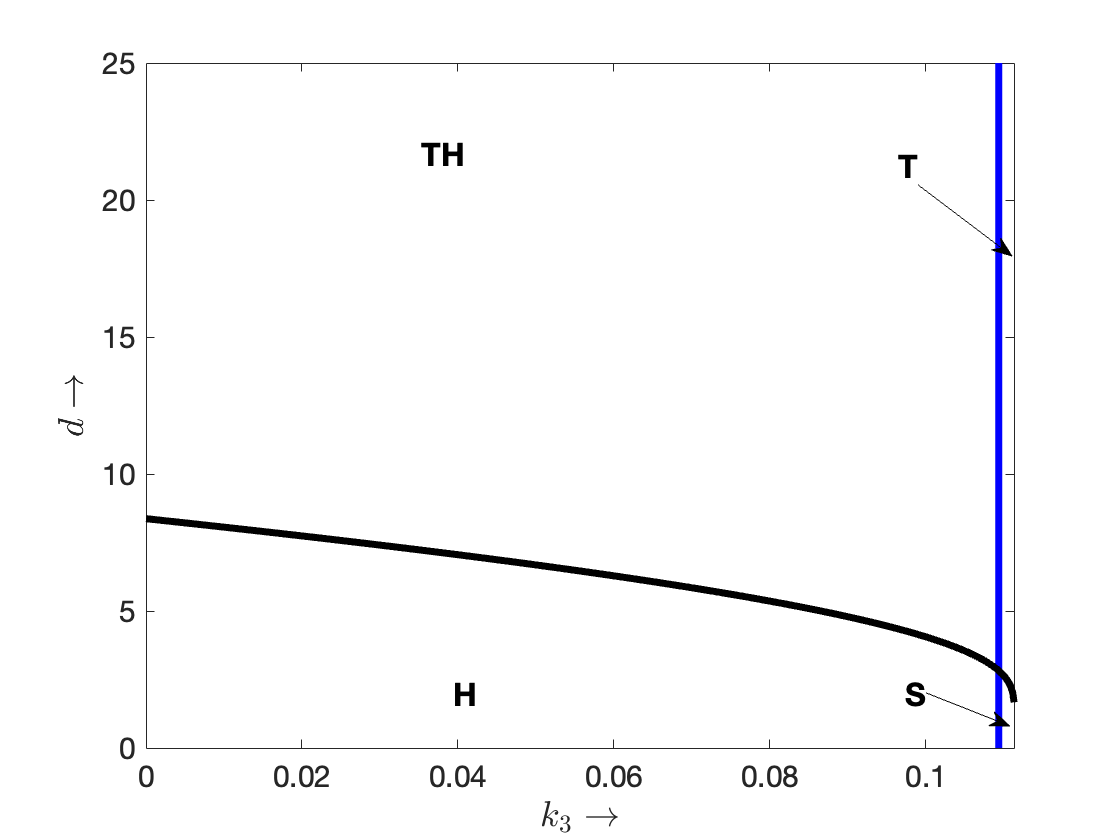}}
\subfigure[]{\includegraphics[width=8.4cm]{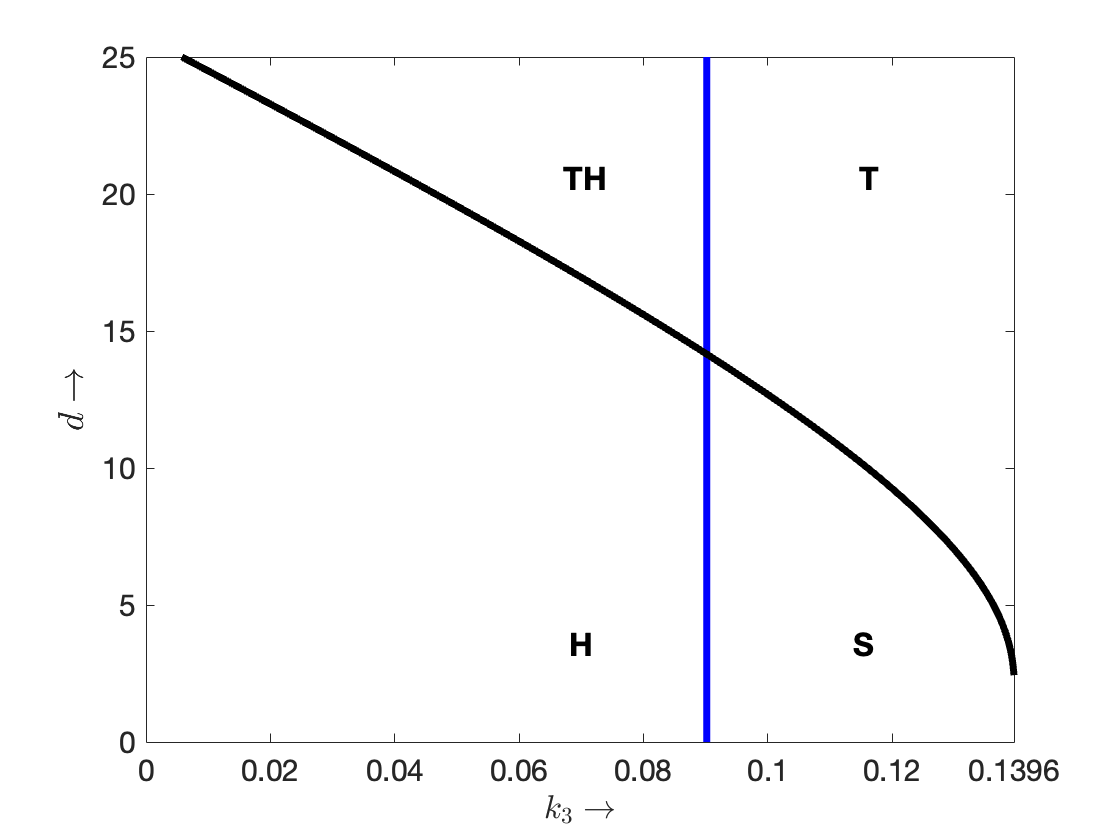}}}
\caption{The plot of the Turing bifurcation curve (black) and Hopf bifurcation threshold (blue) in $(k_3,d)$-parameter space for the system \eqref{alg_eqn_pde} (a) with Holling type II functional response, (a) with Ivlev functional response, for the fixed parameter values $X_0=52$, $r=0.24$, $b_{H}=1.3$, $a_{H}=1$, $b_{I}=1.2367$, $a_{I}=0.4101$.}
\label{Fig:Fig8}
\end{figure}

The study of the temporal model \eqref{algeb_equn} for $X_0 = 52$ reveals that the system \eqref{alg_eqn_pde} can exhibit three coexisting homogeneous steady states for both chosen parameterizations. Additionally, the stable limit cycle disappears in the model with the Holling Type II functional response (see Fig.~\ref{Fig:Fig4}). Under the given parameter constraints, the homogeneous steady state with the highest prey density does not satisfy the Turing instability condition in either model. However, the condition is met for the homogeneous steady state with the lowest prey density in both models. The Turing bifurcation curves and Hopf bifurcation thresholds for both models are shown in Fig.~\ref{Fig:Fig8}. Notably, compared to the previous case (when $X_0 = 23$), the gap between the Turing bifurcation thresholds for the two models is significantly narrower. Furthermore, the stable and Turing domains for the model with the Holling Type II functional response are noticeably smaller than those for the model with the Ivlev functional response within the considered parameter range. The presence of bistability for this parameter setting (unlike the earlier case with $X_0 = 23$) introduces an additional layer of complexity, as the existence of multiple steady states may prevent the system from settling into a stationary pattern.

\begin{figure}[ht]
\centering
\mbox{\subfigure[]{\includegraphics[width=8.1cm]{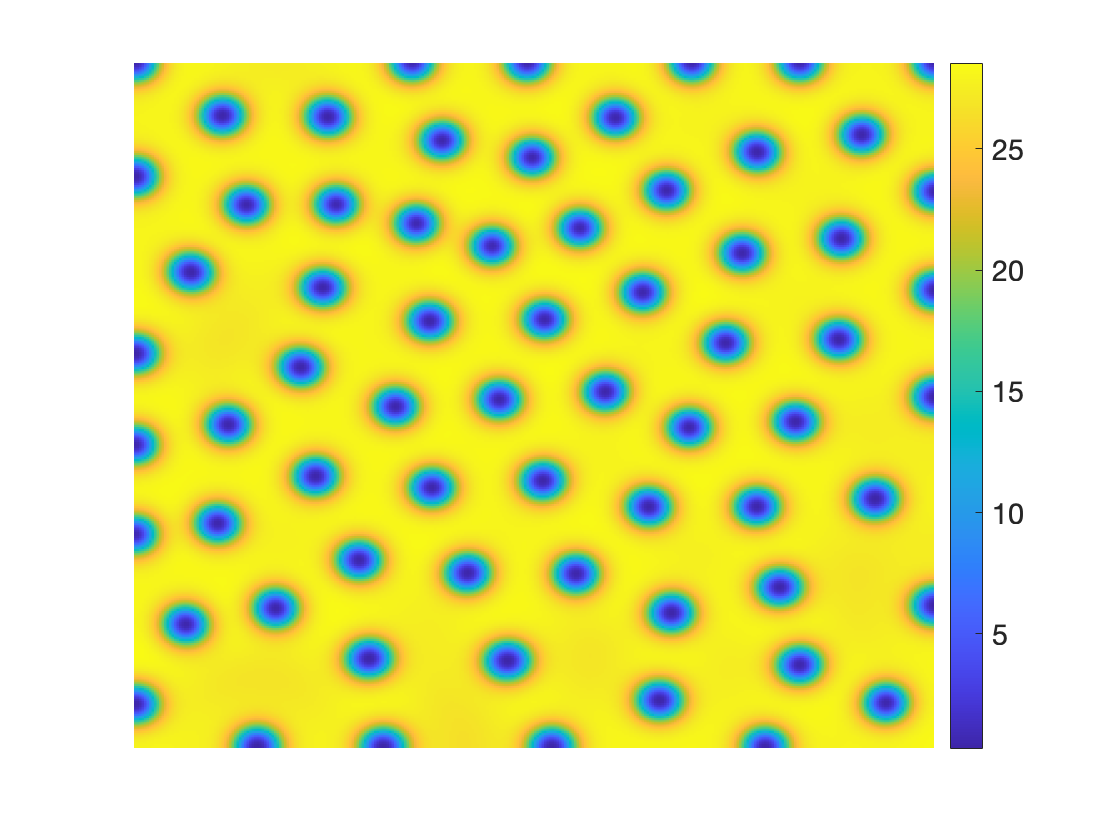}}
\subfigure[]{\includegraphics[width=8.1cm]{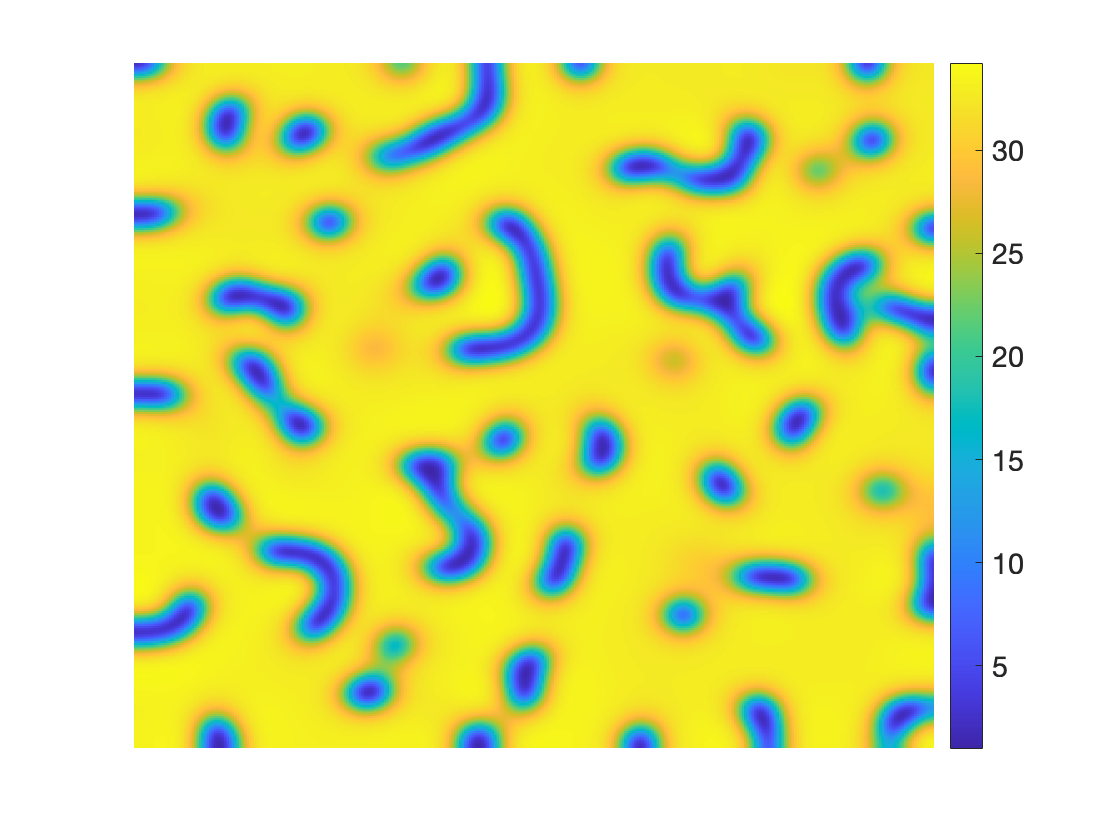}}}
\caption{Patterns observed for the system \eqref{alg_eqn_pde} with (a) Holling Type II functional response and (b) Ivlev functional response, using parameter values $r = 0.24$, $X_0 = 52$, and $k_3 = 0.1$. The parameters specific to each functional response are provided in Table \ref{Table:1}.
}
\label{Fig:Fig9}
\end{figure}

We have analyzed the spatial patterns generated by introducing small heterogeneous perturbations around the Turing-unstable homogeneous steady state, using a predator-to-prey self-diffusion coefficient ratio of $d = 20$. Unlike the earlier case, a lower value of $d$ is considered here. A schematic diagram of all possible patterns is presented in Fig.~\ref{Fig:Fig10}. Interestingly, a new phenomenon has been observed in this setup. For the model with the Holling Type II functional response, a cold-spot pattern emerges for the parameter value $k_3 = 0.1$, and it remains stationary. In contrast, for the model with the Ivlev functional response—despite satisfying the Turing instability condition—a pattern initially forms but evolves into a non-stationary one for the same value of $k_3 = 0.1$. This behavior arises due to the stability of another homogeneous steady state with higher predator density, which disrupts the establishment of a stationary Turing pattern. Additionally, the non-stationary pattern observed in the Ivlev model is a mixture of cold-spot and labyrinthine patterns (Fig. \ref{Fig:Fig9}). Moreover, we have observed the existence of stationary patterns outside the Turing domain, more precisely inside the Hopf domain where the Turing instability condition does not get satisfied, for the model with Ivlev functional response. This phenomenon has not been observed for the model with Holling type II functional response (see Fig. \ref{Fig:Fig10}).

\section{Discussion}{\label{disc}}

Guided by the need to quantify uncertainty in model predictions and to assist ecologists in selecting the most appropriate models, we are interested in exploring alternative model formulations. The primary objective of this work is to investigate the impact of alternative functional responses with similar characteristics in the context of pattern formation. It is crucial to begin with simple ecological models, as they provide a foundational framework for understanding complex interactions in more intricate systems. The Rosenzweig–MacArthur model is unable to produce spatially heterogeneous patterns under self-diffusion alone \cite{MB}. However, the introduction of a density-dependent death rate for the predator population into this model is sufficient to generate various stationary and dynamic patterns under self-diffusion \cite{Shen, IJBC, MB2}. In this study, we examine the structural sensitivity of the Bazykin model and its spatial extension by incorporating alternative formulations of predator grazing behavior. While the Bazykin model traditionally employs the Holling Type II functional response to describe predator grazing, we introduce the Ivlev functional response as an alternative interaction term to study the structural sensitivity of both temporal and spatio-temporal dynamics.

Here, we present preliminary findings on the potential number of steady states, their stability, and bifurcation behavior for the temporal version of a conceptual model incorporating Bazykin-type reaction kinetics and a generalized functional response that satisfies the key characteristics of both the Holling Type II and Ivlev functional responses \cite{Fussmann, seo2018}. Our analytical study reveals that the system can exhibit at most three coexisting steady states, depending on parameter values. It also exhibits two saddle-node bifurcation curves intersecting at a cusp bifurcation point and a Hopf bifurcation curve that collides with one of the saddle-node bifurcation curves at a Bogdanov–Takens bifurcation point in the $(X_0\mbox{-}k_3)$ parameter space, irrespective of the specific choice of functional response. Numerical validation of these analytical findings is carried out using both the Holling Type II and Ivlev functional responses as alternative model formulations. For this purpose, we adopt parameter values from \cite{IJBC} for the Holling Type II case and estimate the Ivlev parameters using least squares approximation. In both models, the Hopf bifurcation curve collides with the upper saddle-node bifurcation curve under the considered parameter restriction. Although the two-parameter bifurcation diagrams share similar structures across both models, a significant shift in their threshold values is observed. This difference notably affects their one-parameter bifurcation diagrams, particularly when varying the interspecies competition rate across different (dimensionless) carrying capacities of the prey population, $X_0$. The analysis of these one-parameter bifurcation diagrams reveals that at higher carrying capacities, the models deviate from similar behavior at the coexisting homogeneous steady state, indicating a divergence in their dynamical outcomes.

Analytical results for the spatio-temporal model are provided to establish the existence and boundedness of solutions for the spatially explicit system. Additionally, we present conditions for the local and global stability of homogeneous steady states under spatially heterogeneous perturbations, the existence of spatially heterogeneous steady states, and the Turing instability of coexisting homogeneous steady states. These analytical findings, along with the conditions required to derive them, are expressed in terms of a generalized functional response. Subsequently, a numerical study is conducted to validate the analytical results and to explore the impact of different functional responses on spatio-temporal pattern formation and Turing instability. The results reveal that in systems with lower carrying capacity, where the one parameter bifurcation diagrams of the models are similar, the difference in Turing bifurcation thresholds is more pronounced. Conversely, in systems with higher carrying capacity, where the one-parameter bifurcation diagrams diverge, the Turing bifurcation thresholds show less variation. Interestingly, stationary patterns have been observed outside the classical Turing domain, particularly within the Hopf domain, for the model incorporating the Ivlev functional response under higher prey carrying capacity. Furthermore, hotspot patterns emerge in the Ivlev-based model when the prey carrying capacity is low. In contrast, the model with the Holling Type II functional response does not exhibit such behaviors under the same parameter conditions.

\begin{figure}[ht]
\centering
\mbox{\subfigure[]{\includegraphics[width=8.8cm]{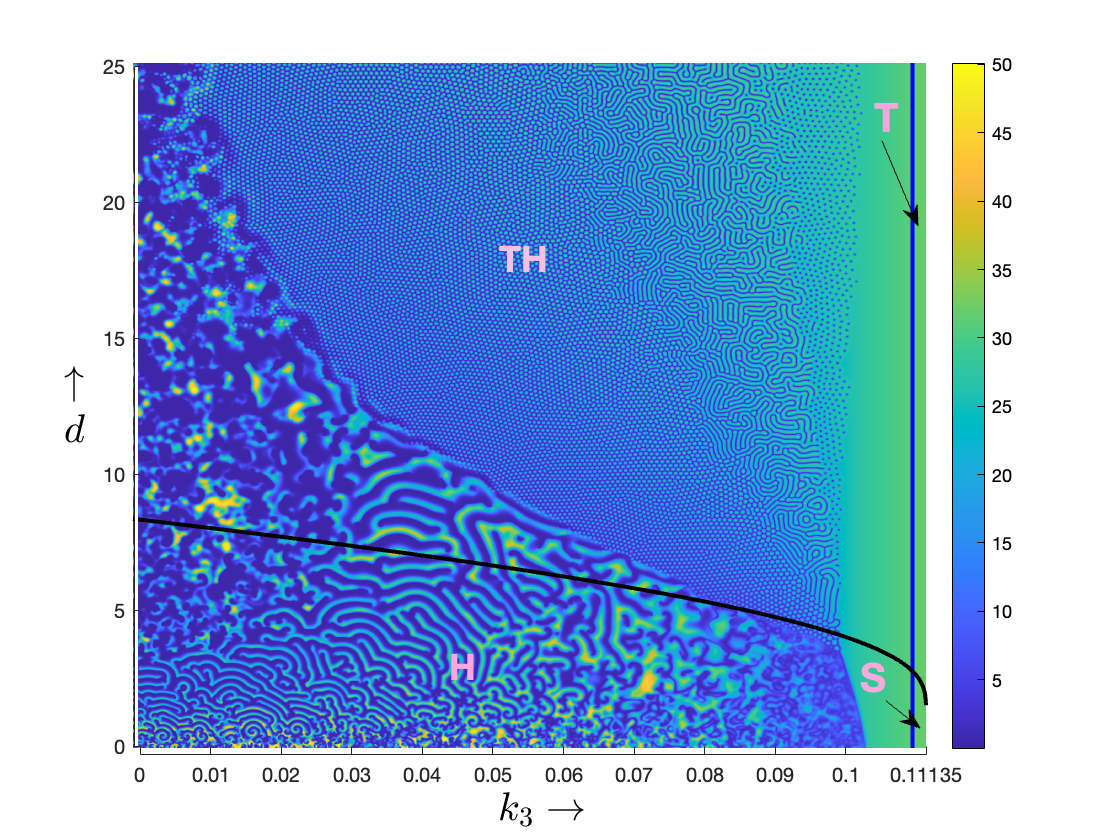}}
\subfigure[]{\includegraphics[width=8.8cm]{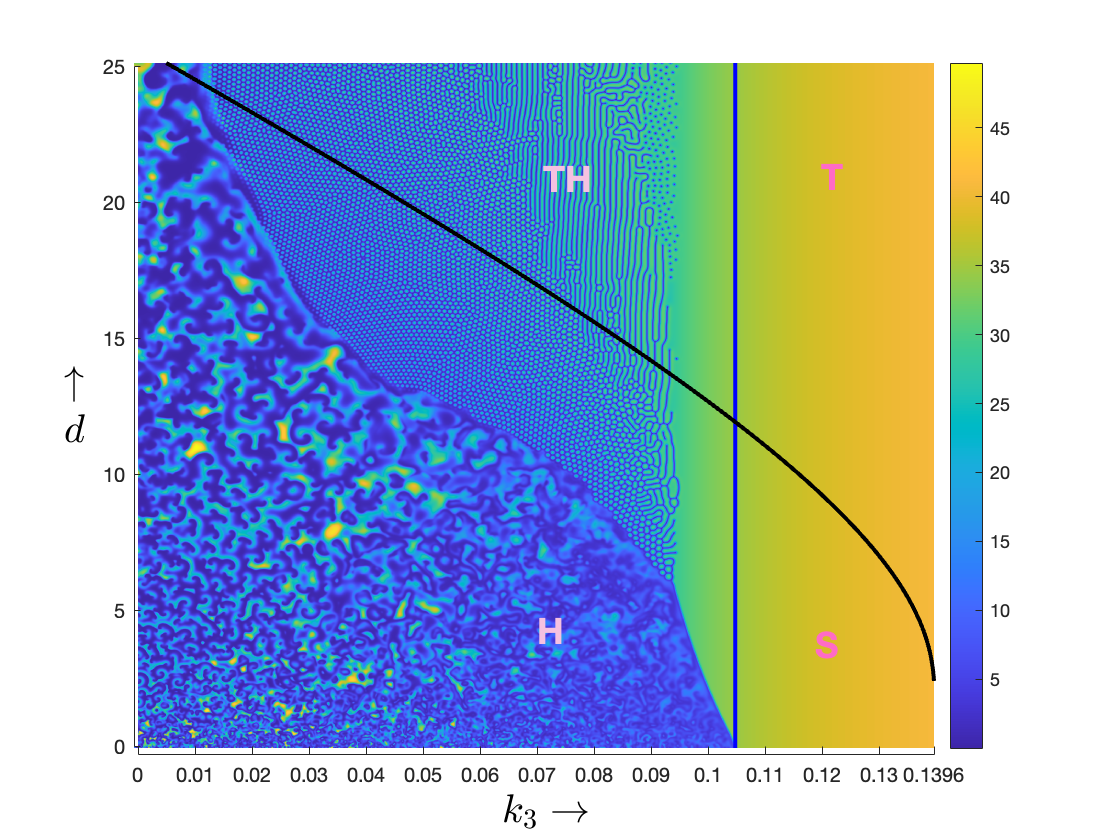}}}
\caption{Schematic diagrams illustrating all possible patterns observed for system \eqref{alg_eqn_pde} are presented for (a) the Holling Type II functional response and (b) the Ivlev functional response. The parameter values corresponding to each functional response are listed in Table \ref{Table:1}. The fixed parameters are set as $r = 0.24$ and $X_0 = 52$. The varying parameters range as follows: $d \in [0, 25]$ and $k_3 \in [0.00001, 0.11135]$ for the model with the Holling Type II functional response, and $k_3 \in [0.00001, 0.1396]$ for the model with the Ivlev functional response.}
\label{Fig:Fig10}
\end{figure}

The mathematical modelling and analysis of the dynamic evolution of interacting populations is an ongoing process. The formulation of a model based on available information and observations about population interactions, followed by systematic mathematical analysis, provides an opportunity to predict the fate of species in terms of survival and patterns of variation. Discrepancies revealed through such analysis often lead to further modifications of the model under consideration. In the context of interacting population models, the functional and numerical responses form the crucial link between two trophic levels, and the parameterization of these responses largely depends on grazing behavior. However, available data on predation and/or interactions between individuals of two species are often insufficient to uniquely parameterize mathematical functions. Recently, Fussmann and Blasius demonstrated that different mathematical representations of the functional response can lead to significantly different dynamic outcomes \cite{Fussmann}. This insight spurred a series of studies aimed at examining the phenomenon in greater depth, with a particular focus on establishing a more rigorous mathematical foundation \cite{seo2018, aldebert2019three, adamson2014defining, Adamson, flora2011structural, aldebert2016structural}. To date, structural sensitivity in population dynamics has been studied mainly in models assuming homogeneous distributions of interacting populations, while such analyses for spatio-temporal models remain rare \cite{MANNA2024134220, Sherratt}. Motivated by our previous work \cite{MANNA2024134220}, we attempt here to provide a detailed mathematical analysis, supported by numerical simulations, to understand the structural sensitivity of stationary and non-stationary patterns in a two-species prey–predator model with Bazykin’s reaction kinetics. Our future goal is to apply this approach to three-species models, as well as to models incorporating other types of functional responses.

\section*{Acknowledgement:} The research of Mr. Indrajyoti Gaine is supported by the PMRF, Ministry of Education, Govt. of  India (Scholar Id: 2303393).



\bibliographystyle{elsarticle-num}
\bibliography{hydrarefuge.bib}

\end{document}